\documentclass[journal]{IEEEtran}
\renewcommand{\baselinestretch}{1.0}
\renewcommand{\theenumii}{\theenumi.\arabic{enumii}}

\usepackage[utf8]{inputenc}
\usepackage{booktabs}
\usepackage{amssymb}
\usepackage{color,soul}
\usepackage[T1]{fontenc}
\usepackage{multicol}
\usepackage{multirow}
\usepackage{textcomp}
\usepackage{multirow}
\usepackage{bigdelim}
\usepackage{graphicx}
\usepackage{epstopdf}
\usepackage{amssymb}
\usepackage{amsmath}
\usepackage{psfrag}
\usepackage{amsthm}
\usepackage[tight,footnotesize]{subfigure}
\usepackage{algorithmic,algorithm}
\usepackage{mathrsfs}
\usepackage{framed}
\usepackage{color}
\usepackage{multirow,enumerate}
\usepackage{bigdelim}
\definecolor{shadecolor}{rgb}{0.92,0.92,0.92}
\usepackage{soul}
\usepackage{color}
\usepackage{cuted}
\usepackage{cite}
\allowdisplaybreaks[4]
\usepackage{xcolor,cite,etoolbox}
\usepackage{url}

\theoremstyle{definition}
\newtheorem{assumption}{Assumption}
\newtheorem{theorem}{Theorem}

\newtheorem{definition}{Definition}

\newtheorem{lemma}{Lemma}

\newtheorem{remark}{Remark}

\makeatletter
\newcommand{\vast}{\bBigg@{3.2}}
\newcommand{\Vast}{\bBigg@{4.5}}

\makeatother

\graphicspath{{Figure/}}

\makeatletter
\pretocmd\@bibitem{\color{black}\csname keycolor#1\endcsname}{}{\fail}
\newcommand\citecolor[1]{\@namedef{keycolor#1}{\color{blue}}}
\makeatother
\begin{document}

\title{
Alternating Direction Method of Multipliers-Based Parallel Optimization for Multi-Agent Collision-Free Model Predictive Control} 
\author{
	Zilong Cheng, 
	Jun Ma, 
	Wenxin Wang, 
 Zicheng Zhu, \\
	Clarence W. de Silva, \IEEEmembership{Life Fellow,~IEEE,}
	and Tong Heng Lee	
\thanks{This work was supported in part by the National Natural Science Foundation of China under Grant 62303390; and in part by the Project of Hetao Shenzhen-Hong Kong Science and Technology Innovation Cooperation Zone under Grant HZQB-KCZYB-2020083. \textit{(Corresponding Author: Jun Ma.)}}
	\thanks{Zilong Cheng, Wenxin Wang, Zicheng Zhu, and Tong Heng Lee are with the Department of Electrical and Computer Engineering; National University of Singapore, Singapore 119077 (e-mail: zilongcheng@u.nus.edu; wenxin.wang@u.nus.edu;  zhuzicheng@u.nus.edu; eleleeth@nus.edu.sg).}
  \thanks{Jun Ma is with the Robotics and Autonomous Systems Thrust, The Hong Kong University of Science and Technology (Guangzhou), Guangzhou, China, also with the Division of Emerging Interdisciplinary Areas, The Hong Kong University of Science and Technology, Hong Kong SAR, China, and also with the HKUST Shenzhen-Hong Kong Collaborative Innovation Research Institute, Futian, Shenzhen, China (e-mail: jun.ma@ust.hk).}
	\thanks{Clarence W. de Silva is with the Department of Mechanical Engineering, University of British Columbia, Vancouver, BC, Canada V6T 1Z4 (e-mail: desilva@mech.ubc.ca).}
\thanks{This work has been submitted to the IEEE for possible publication. Copyright may be transferred without notice,	after which this version may no longer be accessible.}}

	\markboth{}
{Z. Cheng \MakeLowercase{\textit{et al.}}}

\maketitle 

\begin{abstract} 
This paper investigates the collision-free control problem for multi-agent systems. For such multi-agent systems, it is the typical situation where conventional methods using either the usual centralized model predictive control (MPC), or even the distributed counterpart, would suffer from substantial difficulty in balancing optimality and computational efficiency. Additionally, the non-convex characteristics that invariably arise in such collision-free control and optimization problems render it difficult to effectively derive a reliable solution (and also to thoroughly analyze the associated convergence properties). To overcome these challenging issues, this work establishes a suitably novel parallel computation framework through an innovative mathematical problem formulation; and then with this framework and formulation, a parallel algorithm based on alternating direction method of multipliers (ADMM) is presented to solve the sub-problems arising from the resulting parallel structure. Furthermore, an efficient and intuitive initialization procedure is developed to accelerate the optimization process, and the optimum is thus determined with significantly improved computational efficiency. As supported by rigorous proofs, the convergence of the proposed ADMM iterations for this non-convex optimization problem is analyzed and discussed in detail. Finally, a simulation with a group of unmanned aerial vehicles (UAVs) serves as an illustrative example here to demonstrate the effectiveness and efficiency of the proposed approach. 
Also, the simulation results verify significant improvements in accuracy and computational efficiency compared to other baselines, including primal quadratic mixed integer programming (PQ-MIP), non-convex quadratic mixed integer programming (NC-MIP), and non-convex quadratically constrained quadratic programming (NC-QCQP).
\end{abstract}

\def\abstractname{Impact Statement}
\begin{abstract}
\textcolor{black}{Model predictive control (MPC) has demonstrated promising effectiveness in addressing collision-free control problems for multi-agent systems.
However, the inherent non-convexity of optimization problems, computational burden, and performance considerations pose additional challenges to the MPC problem.
To overcome these challenges, the alternating direction method of multiplier (ADMM) serves as an effective optimization technique, which decomposes complex multi-agent collision-free control problems into smaller and more manageable sub-problems. 
With thus this motivation, this work establishes a suitably novel parallel computation framework through an innovative mathematical problem formulation; and then with this framework and formulation, a novel and innovative parallel algorithm based on ADMM is presented to solve the sub-problems arising from the resulting parallel structure.
Furthermore, an efficient and intuitive initialization procedure is developed to accelerate the optimization process, and the optimum is thus determined with significantly improved computational efficiency. 
As supported by rigorous proofs, the convergence of the proposed ADMM iterations for this non-convex optimization problem is analyzed and discussed in detail.
Finally, simulation results clearly verify the effectiveness of the proposed methodologies here; and also showcase the potential for the deployment of these advanced optimization methods to solve the multi-agent collision-free control problem.}
\end{abstract}

\begin{IEEEkeywords}
	Multi-agent system, collision avoidance, parallel computation, model predictive control (MPC), alternating direction method of multipliers (ADMM), non-convex optimization.
\end{IEEEkeywords}

\section{Introduction}

The study on multi-agent collision-free control problems has greatly progressed in recent times, and there already exists a significant corpus of substantial works with that central idea on addressing the control objective of a group of agents to be realized appropriately~\cite{10083278, 9919337, xue2023distributed, 9133452}. 
\textcolor{black}{It is also certainly evident that it would be highly desirable to achieve the control objectives in these multi-agent collision-free control problems suitably effectively, and various approaches have been proposed in the existing literature~\cite{9928314}.}
\textcolor{black}{One commonly used method is the model predictive control (MPC), which has shown promising effectiveness in solving the aforementioned multi-agent collision-free control problems with various dynamic constraints and objective functions.}
Typically, the MPC algorithm can be implemented in centralized or distributed manners. In the centralized MPC~\cite{kayacan2014learning}, the dynamics of all agents, the control objectives, and the environment are modeled as a whole system, and the global optimal control action for each agent is centrally calculated and allocated. One of the main issues of this centralized manner is that the computation burden will be significantly aggravated when the amount of agents is increasing. As a counterpart, distributed algorithms~\cite{9430775,9745747,9464745} \textcolor{black}{are implemented to deal with the aforementioned issue,} with the purpose of formulating a local optimization problem from the global optimization target. It is pertinent to note that, with the distributed algorithms, each agent can explore the information of the neighbor agents and make the control decision independently. Therefore, from the viewpoint of the entire system, all agents can be deployed in a distributed manner, which facilitates effective computation in a parallel fashion. 
However, in such distributed MPC problems, even though the framework exists where the optimum can certainly be reached via coordinating each distributed local optimization problem, there is nevertheless no guarantee that this global optimality is always attained. Tremendous further efforts have thus been dedicated to minimizing the objective gap between the distributed local optimization problems and the global optimization problem~\cite{9254131,9403920}; and additionally, it is revealed in these studies that the performance of the distributed MPC system could suffer significantly when compared to the centralized MPC system approach.

\textcolor{black}{To cater to the aforementioned drawbacks, a significant amount of newer effort has been put into improving the computational efficiency of the optimization schemes for these centralized multi-agent systems, and several more recent works have demonstrated the attainment of satisfying performance.}
One particularly noteworthy direction that researchers have explored focuses on effective and novel problem formulation/re-formulation methodologies~\cite{7865986, he2018neural}. In this approach, an appropriate methodology is developed where it is shown that for a group of specific centralized MPC problems, the optimization problems can be readily formulated in a particular structure, and which can be solved rather efficiently by various effective existing solvers. For example, in~\cite{giselsson2013accelerated}, the dual decomposition technique is utilized to decentralize the multi-agent MPC problem, and then the accelerated proximal gradient (APG) method is used to solve the converted optimization problem in a parallel manner. Another possible way would be to improve the computational efficiency of the MPC solver. Along this line, it is certainly a well-known notion that the traditional MPC problem is essentially a constrained quadratic optimization problem; thus various evident and potential improvements in computational efficiency are still possible (if the information in the optimization problem that has not been otherwise already excavated can be further sufficiently explored and utilized). 
\textcolor{black}{Examples of such efforts can be found in~\cite{wang2022distributed}, where distributed stochastic MPC is presented to reduce the online computational burden by 45\% compared with the centralized counterpart with both using the MOSEK solver. Also, in~\cite{le2020gaussian}, an efficient distributed framework with alternating direction method of multipliers (ADMM) for solving the Gaussian process-based distributed MPC is developed to reduce the computation time by 53\% with the IPOPT (through CasADi) serving as the baseline. However, no collision avoidance or other inequality constraints are involved in~\cite{le2020gaussian}, which are critical and significant impediments in multi-agent collision-free problems.}

\textcolor{black}{Essentially, recent studies on the newer routines possible in the ADMM~\cite{8710607, li2020distributed, luo2019non} point to new and promising possibilities for solving the multi-agent MPC problem.}
The principal idea of the ADMM is to decompose a multi-agent collision-free problem into several small and more manageable sub-problems. With this decomposition, the complicated problem can be split into a couple of subproblems that can be solved in a separate framework efficiently.
ADMM-based algorithms have been applied in many areas such as neural networks~\cite{9597482}, image processing~\cite{9416245}, optimal control~\cite{ma2020arxiv}, and general optimization problem solver~\cite{o2016conic}. Some similar works have been introduced in several existing works. In~\cite{ma2020alternating,zhang2021semi}, the ADMM has been suitably applied in autonomous driving to improve the computational efficiency in motion planning problems.
In~\cite{an2021flexible}, the authors provide a decentralized conflict resolution method for autonomous vehicles based on an extension to the ADMM, and in~\cite{braun2018hierarchical}, a distributed MPC method via ADMM is proposed for power networks system. \textcolor{black}{However, in both works, no non-convex constraints and parallel structure are considered in the conventional ADMM framework. In this sense, addressing the issue of non-convexity as typically encountered in multi-agent collision-free problems, while employing the ADMM with a parallel structure, remains an open and intriguing question.}

\textcolor{black}{With the above descriptions as a backdrop, 
the non-convexity of the optimization, the trade-off between optimality and computational efficiency, and control performance are major concerns in multi-agent collision-free control problems.}
To overcome these challenges, this paper proposes an ADMM-based parallel optimization methodology for solving the centralized MPC problem with collision-avoidance constraints.
The merits present in both the distributed method and the centralized counterpart are thus suitably combined, while their major shortcomings are effectively avoided. 
Specifically, the ADMM is employed to solve the centralized collision avoidance MPC problem in a parallel manner, resulting in a significant improvement in the computational efficiency (over ten times faster than the log barrier method). 
Then, the centralized collision avoidance MPC problem is divided into two groups of sub-problems. 
The first group deals with the quadratic objective function, dynamic constraints, and box constraints for each UAV system. 
The second group deals with the single collision avoidance constraint. 
Leveraging the ADMM, the optimization problem is mathematically converted to a consensus optimization problem. 
Notably, the computational complexity of the first group of parallel sub-problems is much heavier than the computational complexity of the centralized sub-problem. Consequently, the centralized sub-problem with the collision avoidance constraint can be efficiently solved without considering the other constraints.
After the two groups of problems are solved in an iteration, the solutions are shared among all UAV systems, marking the commencement of the next iteration.

\textcolor{black}{The main contributions of this paper are listed as follows: (1) By introducing an innovative mathematical problem formulation, a parallel computation framework is established, enabling the efficient application of advanced optimization methods for solving the centralized model predictive collision-free control problem.
(2) Leveraging the parallelization structure, the proposed methodology employs the ADMM algorithm to solve each sub-problem, resulting in a significant improvement in computational efficiency.
(3) Despite the inherent non-convexity of the optimization problem, an efficient and intuitive initialization procedure is developed to expedite the optimization process, thereby facilitating straightforward determination of the optimum. Furthermore, rigorous proofs are provided to analyze the convergence of the proposed algorithm for the formulated non-convex optimization problem.}

The remainder of this paper is organized as follows. Section II formulates the multi-agent collision-free control problems. The dynamic constraints, and also the collision avoidance constraints, are characterized; and subsequently, a non-convex optimization problem is introduced. Section III proposes the use of the ADMM algorithm to solve the formulated problem in a parallel manner. In section IV, the convergence analysis of the proposed non-convex ADMM iterations is given. In Section V, a multi-agent system with a group of unmanned aerial vehicles (UAVs) is introduced as an example to test out (and thus validate) the proposed methodology; and comparisons between the existing methods and the proposed method are carried out with supporting simulations. Finally, pertinent conclusions of this work are provided in Section VI.

\section{Problem Statement}
\subsection{Notations}
Before formulating the main problem, we introduce the following notations at first, which will be used in the remaining text. $\operatorname{diag}(\cdot)$ denotes a block diagonal operator. $A^\top$ denotes a transpose of the matrix $A$. $\mathbb R^{m\times n}$  $(\mathbb R^n)$ denotes a real matrix with $m$ rows and $n$ columns (real column vector with the dimension $n$). $\mathbb S^n$ denotes a symmetric matrix space with the dimension $n$. $\mathbb S^n_+$ $(\mathbb S_{++}^n)$ denotes a real positive semi-definite (positive definite) symmetric matrix with the dimension $n$. Vector combination notation with respect to column is denoted by $(\cdot)$, i.e., $a = (a_1,a_2,\dotsm,a_N)=\left[a_1^\top,a_2^\top,\dotsm,a_N^\top\right]^\top$. Operator combination notation is also represented by $(\cdot)$, such as the operator $f$ including $f_1,f_2,\dotsm,f_N$ can be denoted by $f=(f_1,f_2,\dotsm,f_N)$. $I_n$ represents the identity matrix with the dimension $n\times n$. The operator $\operatorname{Tr}(A)$ denotes a trace of the square matrix $A$. The operator $\langle A, B \rangle$ denotes a Frobenius inner product, i.e., $\langle A,B\rangle= \operatorname{Tr}\left(A^\top B\right)$ for all $A,B \in \mathbb R^{m\times n}$. The norm operator based on the inner product operator is defined by $\|x\|=\sqrt{\langle x,x\rangle}$ for all $x\in \mathbb R^{m\times n}$. The operator $\|x\|_1$ denotes the $\ell_1$ norm of the vector $x$. $\succeq$ denotes the generalized inequality symbol (element-wisely greater or equal than).

\subsection{Dynamic Constraints}
\textcolor{black}{In this subsection, the formulation of the multi-agent MPC problem is given. In the MPC framework, the control problem is formulated as an optimization problem subject to the dynamic constraints and other various constraints, }
\textcolor{black}{where the obtained optimal solution is a sequence of control inputs at each time stamp.} 
To solve the MPC problem, a precise system model is usually necessary for improving the control performance. Thus, the dynamic constraints are the essential part of an MPC controller.

 Before analyzing the multi-agent system, we investigate the problem of the single agent system as first. Here, for all $\tau=t,t+1,\dotsm,t+T-1$, the linear time-variant (LTI) system dynamics can be represented by
\begin{IEEEeqnarray}{rl}\label{eq:system}
x(\tau+1)=Ax(\tau)+Bu(\tau),
\end{IEEEeqnarray}
where $A\in\mathbb R^{n\times n}, B\in\mathbb R^{n\times m}$ are the system state and control input matrices, respectively; $x(\tau)\in\mathbb R^n,u(\tau)\in\mathbb R^{m}$ are the system state and control input vectors, respectively; $t$ is the initial time stamp; $T$ represents the prediction horizon. We assume that the weighting parameters in terms of the system state and system input variables are time-invariant. To track a given reference signal, the tracking optimization problem of each individual agent with a given quadratic objective function can be formulated as
\begin{IEEEeqnarray*}{l}\label{eq:opt_individual_1}
\min_{x(\tau),u(\tau)} \quad
\big(x(t+T)-r(t+T)\big)^\top\hat Q_T\big(x(t+T)-r(t+T)\big)\\
+\displaystyle\sum_{\tau =t}^{t+T-1} \Big(\big(x(\tau)-r(\tau)\big)^\top\hat Q\big(x(\tau)-r(\tau)\big)+u(\tau)^\top\hat Ru(\tau)\Big)\\
\operatorname{subject\ to} \quad x(\tau+1)=Ax(\tau)+Bu(\tau)\\
\,\,\,\,\quad\quad\quad\quad\quad\tau=t,t+1,\dotsm,t+T-1, \IEEEyesnumber
\end{IEEEeqnarray*}
where $\hat Q\in\mathbb S^{n}_+$ and  $\hat R\in\mathbb S^{m}_{+}$ are the weighting matrices for the system state and system input variables, respectively; $r(\tau)$ is the reference signal at the time $\tau$. 


Since the initial system state variable $x(t)$ cannot be influenced by the optimization result, the optimization variables $x$ and $u$ in terms of the system state and system input variables, respectively, are defined as
\begin{IEEEeqnarray*}{rCl}
x&=&\big(x(t+1),x(t+2),\dotsm,x(t+T)\big)\in \mathbb R^{nT}\\
u&=&\big(u(t),u(t+1),\dotsm,u(t+T-1)\big)\in\mathbb R^{mT}, \yesnumber
\end{IEEEeqnarray*}
and the initial state variable is redefined as $x^t=x(t)$ to eliminate the confusion of notation in the remaining text. Then, it follows that
\begin{IEEEeqnarray*}{l}
x=Gx^t+Hu,\yesnumber
\end{IEEEeqnarray*}
where
\begin{IEEEeqnarray*}{rCl}
G&=&\Big(A,A^2,\dotsm,A^{T-1},A^T\Big)\in\mathbb R^{nT\times n}\\
H&=&\begin{bmatrix}
B&0&0&\dotsm&0\\
AB&B&0&\dotsm&0\\
A^2B&AB&B&\dotsm&0\\
\vdots&\vdots&\vdots&\ddots&\vdots\\
A^{T-1}B&A^{T-2}B&A^{T-3}B&\dotsm&B
\end{bmatrix}\in\mathbb R^{nT\times mT}.\\\yesnumber
\end{IEEEeqnarray*}
Then, the problem~\eqref{eq:opt_individual_1} can be converted into a compact form:
\begin{IEEEeqnarray*}{rl}\label{eq:opt_individual_2}
\min_{x,u} \quad
& (x-r)^\top Q(x-r)+u^\top Ru\\
\operatorname{subject\ to} \quad &
x=Gx^t+Hu,\yesnumber
\end{IEEEeqnarray*}
where
\begin{IEEEeqnarray*}{rCl}
r&=&\big(r(t+1),r(t+2),\dotsm,r(t+T)\big)\in\mathbb R^{nT}\\
Q&=&\operatorname{diag}\Big(\underbrace{\hat Q,\hat Q,\dotsm,\hat Q}_{T-1},
\hat Q_T \Big)\in\mathbb R^{nT\times nT}\\
R&=&\operatorname{diag}\Big(\underbrace{\hat R,\hat R,\dotsm,\hat R}_{T}\Big)\in\mathbb R^{mT\times mT}.\yesnumber
\end{IEEEeqnarray*}

Additionally, as an inevitably encountered challenge, most of the systems in the real world have the characteristic of the dead zone, but it is usually not represented in the system model because of the challenge to identify. When the system model without the dead zone is utilized in the MPC problem, some terms of the obtained control input signal in small value will be covered by the dead zone in practical use and the performance will be thus diminished. To avoid this issue, a lasso term with regard to control input signal $\Vert u\Vert_1$ is added in the objective function to penalize the sparsity of the control input signal. As a result, if the magnitude of a control input signal is very small, then adding a lasso term to the objective function can force these small control input signals to zero, and thus the performance of the optimization results will be further improved~\cite{pakazad2013sparse}. It is also pertinent to note that, compared to the quadratic regularization term $\Vert u\Vert_2^2$, the lasso term can perfectly achieve this target. Even though the quadratic regularization term can reduce the value of the control input, it cannot force the small values to be zeros.

Besides, as a common practice, the limits of the system state variables and system input variables are also taken into consideration, which are represented as box constraints in the MPC problem. Then, it follows that the problem~\eqref{eq:opt_individual_2} can be generalized and represented by
\begin{IEEEeqnarray*}{rl} \label{eq.mpc1}
\min_{x,u} \quad
& (x-r)^\top Q(x-r)+u^\top Ru+\alpha\|u\|_1\\
\operatorname{subject\ to} \quad
&x=Gx^t+Hu\\
&x\in \mathcal X,u\in\mathcal U,\yesnumber
\end{IEEEeqnarray*}
where $\mathcal X$ and $\mathcal U$ denote the box constraints of the system state and system input variables, respectively. The box constraints can be denoted by $\mathcal X=\{x\in\mathbb R^{nT}\;|\; \underline x\preceq x\preceq \overline x\}$, where $\underline x\in\mathbb R^{nT},\, \overline x\in\mathbb R^{nT}$ denote the minimum value and maximum value of the system state variables, respectively; $\mathcal U=\{u\in\mathbb R^{mT}\,|\, \underline u\preceq u\preceq \overline u\}$, where $\underline u\in\mathbb R^{mT},\, \overline u\in\mathbb R^{mT}$ denote the minimum value and maximum value of the system input variables, respectively; $\alpha$ is the weighting parameter with respect to the lasso term.

To further simplify the optimization problem, we introduce the definition of the indicator function. Specifically, the nonlinear constraints in the optimization problem could be written into indicator functions; and thus these constraints could be included in the objective function. And then, the optimization problem can be further split into two groups of sub-problems. This effect will be introduced in the next section in detail. At this stage, we give the definition of the indicator function first.
\begin{definition}
	The indicator function with respect to a cone $\mathbb B$ is defined as
	\begin{IEEEeqnarray*}{l}
	\delta_\mathbb B(B)=\left\{\begin{array}{ll} 0 & \text{if } B\in \mathbb B\\
	\infty & \text{otherwise}.
	\end{array}\right.\yesnumber
	\end{IEEEeqnarray*}
\end{definition}

By using the indicator function, we can convert the optimization problem~\eqref{eq.mpc1} to
\begin{IEEEeqnarray*}{rl}
\min_{x,u} \quad
& (x-r)^\top Q(x-r)+\delta_\mathcal X(x)+u^\top Ru\\
&+\delta_\mathcal U(u)+\alpha\|u\|_1\\
\operatorname{subject\ to} \quad
&x=Gx^t+Hu.\yesnumber
\end{IEEEeqnarray*}

After the formulation of the optimization problem in terms of an individual agent, we generalize the single agent MPC problem to a multi-agent optimization problem with $N$ agents. Each agent is required to satisfy the aforementioned dynamic constraint and box constraints. Therefore, the multi-agent MPC optimization problem is represented by
\begin{IEEEeqnarray*}{rl} \label{eq.multi}
\min_{x_i,u_i} \quad
& \displaystyle \sum_{i=1}^N \Big((x_i-r_i)^\top Q_i(x_i-r_i)+\delta_{{\mathcal X}_i}(x_i)\\
&\quad\quad+u^\top_iR_iu_i+\delta_{{\mathcal U}_i}(u_i)+\alpha_i\|u_i\|_1\Big)\\
\operatorname{subject\ to} \quad
&x_i=G_ix_{i}^t+H_iu_i\\
&\forall i=1,2,\dotsm,N,\IEEEyesnumber
\end{IEEEeqnarray*}
where the subscript $i$ of each variable and parameter denotes that the corresponding variable and parameter belong to the $i$th agent. 

\begin{remark}
\textcolor{black}{The agents considered in this paper operate in a collaborative manner. Additionally, the broadcast communication model is employed for the autonomous agents, i.e., each agent broadcasts the information to its neighbors via the broadcaster and receives the information from its neighbors via the receiver at each 
time stamp.}
\end{remark}
\subsection{Collision Avoidance Constraint}
To deal with the multi-agent collision-free control problem, we introduce the collision avoidance constraints into the optimization problem. The collision avoidance constraints can be represented in several different formats in the existing literature. The most basic one is to define the safe area as a norm ball~\cite{chen2019autonomous}. Then, the collision avoidance constraint can be represented by maintaining the distance to the obstacle to be greater than a specific value. However, this kind of constraint is naturally non-convex, which impedes the implementation of effective algorithms to find a feasible solution. Taking the aforementioned issue into account, mixed integer programming (MIP) is utilized to solve the collision avoidance problem in an effective manner~\cite{8966482,9981138,9922287}. \textcolor{black}{The MIP-based collision avoidance constraint offers a closed-form solution, which serves as a relaxation of the original norm ball collision avoidance constraint.} More specifically, an intuitive explanation of the MIP collision avoidance constraints is similar to relaxing a ball constraint to a box constraint, i.e., finding a minimal box that contains the given ball. In terms of this kind of constraint, even though the MIP problem is non-convex, a feasible point with a satisfying duality gap can be easily found. Besides, considering the intuitive selection of the possible solution, the computational efficiency can be further improved. Thus, we consider the use of MIP-based collision avoidance constraints in our proposed algorithm.

Here, to clearly describe the formulation of the MIP-based collision avoidance problem, we first introduce a case in the 2D plane as an example, and then we extend the description to a 3D scenario in this work.

Consider two agents with the position $(x_1,y_1)$ and $(x_2,y_2)$, respectively. The safe distance is given as $d$, and the collision avoidance constraint can be described as $(x_1-x_2)^2+(y_1-y_2)^2\ge d^2$. If a minimal box constraint is considered as the collision avoidance constraint, then the constraints can be described in a linear form, that is $(x_1-x_2)\ge d-\varepsilon c_1$, $(x_2-x_1)\ge d-\varepsilon c_2$, $(y_1-y_2)\ge d-\varepsilon c_3$, $(y_2-y_1)\ge d-\varepsilon c_4$, and $\sum_{j=1}^4 c_j\le 3$, where $\varepsilon$ is chosen as a large number, and $c_1,c_2,c_3,c_4$ are four binary variables belonging to the set $\{0,1\}$. The relaxed constraint ensures that each agent is definitely outside the minimal box constraint of the other agent. If the binary variable value is $1$, it means the corresponding linear inequality constraint is trivial. To ensure at least one nontrivial linear inequality constraint exists, i.e., two agents keep a safe distance, the summation of all binary variable values must be less than the total number of the binary variables.

\textcolor{black}{In terms of two agents with the 3D Cartesian coordinates $(x_1,y_1,z_1)$ and $(x_2,y_2,z_2)$, the collision avoidance constraint can be described as $(x_1-x_2)^2+(y_1-y_2)^2+(z_1-z_2)^2\ge d^2$. Similarly, the constraints can be described in a linear form, that is $(x_1-x_2)\ge d-\varepsilon c_1$, $(x_2-x_1)\ge d-\varepsilon c_2$, $(y_1-y_2)\ge d-\varepsilon c_3$, $(y_2-y_1)\ge d-\varepsilon c_4$, $(z_1-z_2)\ge d-\varepsilon c_5$, $(z_2-z_1)\ge d-\varepsilon c_6$, and $\sum_{j=1}^6 c_j\le 5$, where $c_1,\dots,c_6$ are binary variables belonging to the set $\{0,1\}$. Therefore, to formulate the collision avoidance problem in a general case, for all positive integers $i=1,2,\dotsm,N$, $j=2,3,\dotsm,N,j>i$, $\tau=t,t+1,\dotsm,t+T-1$, and $\ell=1,2,\dotsm,6$, we have the following MIP constraints:
\begin{IEEEeqnarray}{l}
E_{\tau\ell}(x_i-x_j)\ge d-\varepsilon c_{ij\tau\ell}\IEEEyesnumber\IEEEyessubnumber\label{eq:MILP}\\
\sum_{\ell=1}^6c_{ij\tau \ell}\le 5,\IEEEyessubnumber\label{eq:MILP_sum}
\end{IEEEeqnarray}
where the linear operator $E_{\tau \ell},\forall \ell=1,3,5$ extracts the position state variables in 3 dimensions $x,y,z$ in the Cartesian coordinates at the time stamp $\tau$, respectively;  the linear operator $E_{\tau \ell},\forall \ell=2,4,6$ extracts the position state variables in 3 dimensions $x,y,z$ in Cartesian coordinates at the time stamp $\tau$ with a negative sign, respectively;} $d$ denotes the safe distances in all three dimensions $x,y,z$, respectively; $\varepsilon$ is a large number given in the MIP problem; $c_{ij\tau \ell}$ represents the integer optimization variables, which introduce the non-convex constraints into the optimization problem. The inequality constraint~\eqref{eq:MILP_sum} provides the condition that there is no such a case where all the inequality constraints in~\eqref{eq:MILP} are activated.

Then, we define the vector $c\in\mathbb R^{3N(N-1)T}$ as a vector including the $c_{ij\tau \ell}$ for all $i=1,2,\dotsm,N$, $j=2,3,\dotsm,N,j>i$, $\tau=t,t+1,\dotsm,t+T-1$, $\ell=1,2,\dotsm,6$ in the ascending order, that is
\begin{IEEEeqnarray*}{rCl}
c&=&\Big(c_{12t1},\dotsm,c_{12t6},\dotsm,\\
   &&c_{12(t+T-1)1},\dotsm,c_{12(t+T-1)6},\dotsm,\\
   &&c_{(N-1)Nt1},\dotsm,c_{(N-1)Nt6},\dotsm,\\
   &&c_{(N-1)N(t+T-1)1},\dotsm,c_{(N-1)N(t+T-1)6}\Big).\yesnumber
\end{IEEEeqnarray*}

Then, the optimization problem is equivalent to
\begin{IEEEeqnarray*}{rl} \label{eq.mpc2}
\min_{x_i,u_i,c} \quad
& \displaystyle \sum_{i=1}^N \Big((x_i-r_i)^\top Q_i(x_i-r_i)+\delta_{{\mathcal X}_i}(x_i)\\
&+u^\top_iR_iu_i+\delta_{{\mathcal U}_i}(u_i)+\alpha_i\|u_i\|_1\Big)+\delta_{\mathcal C}(c)\\
\operatorname{subject\ to} \quad
&x_i=G_ix_{i}^t+H_iu_i,\\
&E_{\tau \ell}(x_i-x_j)- d+\varepsilon c_{ij\tau \ell}\ge0\\
&\displaystyle\sum_{\ell=1}^6c_{ij\tau \ell}- 5\le 0\\
&\forall i=1,2,\dotsm,N\\
&\forall j= 2,3,\dotsm,N,j>i\\
&\forall \tau=t,t+1,\dotsm,t+T-1\\
&\forall \ell=1,2,\dotsm,6,\yesnumber
\end{IEEEeqnarray*}
where $\mathcal C$ denotes a non-convex cone where only two elements $0,1$ exist, i.e., $\mathcal C=\{0,1\}^{3N(N-1)T}$. To simplify the notations, define the vector $y=(x_1,u_1,x_2,u_2,\dotsm,x_N,u_N,c)\in\mathbb R^{(m+n+3N-3)TN}$, and the sets $\mathcal Y_1$ and $\mathcal Y_2$ as
\begin{IEEEeqnarray*}{l}
\mathcal Y_1=\left\{\left.y=\begin{bmatrix}x_1\\u_1\\x_2\\u_2\\\vdots\\x_N\\u_N\\c\end{bmatrix}\,\right|\, 
\begin{array}{l}
H_iu_i-x_i+G_ix_{i}^t=0\\\forall i=1,2,\dotsm,N
\end{array}
\right\}\\
\mathcal Y_2=\left\{\left.y=\begin{bmatrix}x_1\\u_1\\x_2\\u_2\\\vdots\\x_N\\u_N\\c\end{bmatrix}\,\right|\, 
\begin{array}{l} 
E_{\tau \ell}(x_i-x_j)- d+\varepsilon c_{ij\tau \ell}\ge0\\
\displaystyle\sum_{\ell=1}^6c_{ij\tau \ell}- 5\le 0\\
\forall i=1,2,\dotsm,N\\
\forall j=2,3,\dotsm,N,j>i\\
\forall \tau=t,t+1,\dotsm,t+T-1\\
\forall \ell=1,2,\dotsm,6
\end{array}\right\}.\\\yesnumber
\end{IEEEeqnarray*}
\begin{remark}
	To realize the parallel computation, we separate the independent constraints and the coupling constraints in the multi-agent system by defining two sets $\mathcal Y_1$ and $\mathcal Y_2$. The set $\mathcal Y_1$ contains only the constraint belonging to the individual agent, while the set $\mathcal Y_2$ contains only the coupling constraints. And it is straightforward to see that the sub-problem subject to $\mathcal Y_1$ can be calculated in a parallel manner, since each agent is independent in this sub-problem.
\end{remark}
To denote the optimization problem in a compact form and convert it to a separatable problem, we define the function $f_i$ as
\begin{IEEEeqnarray*}{rCl}
f_i(x_i,u_i)&=&(x_i-r_i)^\top Q_i(x_i-r_i)+\delta_{{\mathcal X}_i}(x_i)\\
&&+u^\top_iR_iu_i+\delta_{{\mathcal U}_i}(u_i)+\alpha_i\|u_i\|_1.\yesnumber
\end{IEEEeqnarray*}
Then, we define the operator $f:\mathbb R^{(m+n+3N-3)TN}\rightarrow\mathbb R^{N+1}$ and $g:\mathbb R^{(m+n+3N-3)TN}\rightarrow\mathbb R^{N+1}$, and also we introduce a new vector $z\in\mathbb R^{(m+n+3N-3)TN}$, such that
\begin{IEEEeqnarray*}{rCl}
f(y)&=&\big(f_1,f_2,\dotsm,f_N,\mathbf 0\big)(y),\\
g(z)&=&\big(\underbrace{\mathbf 0,\mathbf 0,\dotsm,\mathbf 0}_{N},\delta_\mathcal C\big)(z),\yesnumber
\end{IEEEeqnarray*}
where $\mathbf 0(\cdot)$ denotes the zero operator in the Hilbert space. Finally, the multi-agent optimization problem~\eqref{eq.multi} is formulated as
\begin{IEEEeqnarray*}{rl}\label{eq:outer_admm_prob}
\min_{y,z} \quad
& \mathbf 1^\top f(y)+\delta_{\mathcal Y_1}(y)+\mathbf 1^\top g(z)+\delta_{\mathcal Y_2}(z)\\
\operatorname{subject\ to} \quad
&y-z=0,\yesnumber
\end{IEEEeqnarray*}
where $z\in\mathbb R^{(m+n+3N-3)TN}$ is the consensus variable; $\mathbf 1$ denotes the vector with an appropriate dimension and all entries as $1$.

\section{Optimization for Parallel MPC}
The main purpose of the proposed optimization scheme is to separate the centralized MPC problem into two groups of sub-problems. The first group deals with the quadratic objective function, dynamic constraints, and box constraints for each UAV system (i.e., the individual constraints), and thus the first group of sub-problems can be calculated in a parallel manner. The second group deals with the collision avoidance constraint (i.e., the coupling constraints), which is usually straightforward to be solved. Therefore, we propose the use of the ADMM optimization scheme. To solve the problem by ADMM, we firstly introduce the augmented Lagrangian function, which can be denoted by
\begin{IEEEeqnarray*}{rCl}
\mathcal L_\sigma(y,z;\lambda)&=&\mathbf 1^\top f(y)+\delta_{\mathcal Y_1}(y)\\
&&+\mathbf 1^\top g(z)+\delta_{\mathcal Y_2}(z)+\langle\lambda,y-z\rangle+\dfrac{\sigma}{2}\|y-z\|^2\\
&=&\mathbf 1^\top f(y)+\delta_{\mathcal Y_1}(y)+\mathbf 1^\top g(z)+\delta_{\mathcal Y_2}(z)\\
&&+\dfrac{\sigma}{2}\|y-z+\sigma^{-1}\lambda\|^2-\dfrac{1}{2\sigma}\|\lambda\|^2,\yesnumber
\end{IEEEeqnarray*}
where $\lambda=\big(\lambda_1,\lambda_2,\dotsm,\lambda_N,\lambda_c\big)\in\mathbb R^{(m+n+3N-3)TN}$ is the Lagrangian multiplier; $\sigma$ is the penalty parameter of ADMM. Then, the ADMM algorithm can be denoted in an iterative computation manner
\begin{IEEEeqnarray*}{rCl}\label{eq:ADMM_iterations}
y^{k+1}&=&\displaystyle\operatorname*{argmin}_y \;\mathcal L_\sigma\Big(y,z^k;\lambda^k\Big)\IEEEyesnumber\IEEEyessubnumber\label{eq:ADMM_iterations_1}\\
z^{k+1}&=&\displaystyle\operatorname*{argmin}_z \;\mathcal L_\sigma\Big(y^{k+1},z;\lambda^k\Big)\IEEEyessubnumber\label{eq:ADMM_iterations_2}\\
\lambda^{k+1}&=& \lambda^k+\sigma\Big(y^{k+1}-z^{k+1}\Big),\IEEEyessubnumber
\end{IEEEeqnarray*}
where the superscript $k$ denotes the optimization variable in the $k$th iteration. To test the attainment of the optimum, a stopping criterion for the iterations is needed. Here, we define the primal residual error as $\epsilon_{pri}^k=\|y^k-z^k\|$, and the stopping criterion can be chosen as $\epsilon_{pri}^k\le\hat \epsilon$, where $\hat \epsilon$ is a small positive value. It means that the iterative computation will stop when $\epsilon_{pri}^k\le\hat \epsilon$ is satisfied and the optimum is attained.

\subsection{First Sub-Problem in ADMM}
The first sub-problem~\eqref{eq:ADMM_iterations_1} in the ADMM iterations can be calculated in a parallel manner, which means that for each agent $i$, we have
\begin{IEEEeqnarray*}{l}
\Big(x^{k+1}_i,u^{k+1}_i\Big)=\displaystyle\operatorname*{argmin}_{x_i,u_i} \;\bigg\{f_i(x_i,u_i)+\frac{\sigma}{2} \Big\|(x_i,u_i)\\
\qquad-z^k_i+\sigma^{-1}\lambda_i^k\Big\|^2\;\bigg|\;\;H_iu_i-x_i+G_ix_{i}^t=0\bigg\},\yesnumber
\end{IEEEeqnarray*}
where $z$ is separated and denoted by $z=(z_1,z_2,\dotsm,z_N,z_c)$ and $z\in\mathbb R^{(m+n)T\times (m+n)T\times \dotsm (m+n)T \times 3N(N-1)T}$. It follows that for each agent $i$, the following optimization problem is solved:
\begin{IEEEeqnarray*}{rl}\label{eq:opt_second_pro_2}
\min_{x_i,u_i} \quad
& (x_i-r_i)^\top Q_i(x_i-r_i)+u_i^\top R_iu_i\\
&+\alpha_i\|u_i\|_1+\dfrac{\sigma}{2}
\Big\|(x_i,u_i)-z^k_i+\sigma^{-1}\lambda_i^k\Big\|^2\\
\operatorname{subject\ to} \quad
&H_iu_i-x_i+G_ix_{i}^t=0\\
&\underline x_i\preceq x_i\preceq \overline x_i\\
&\underline u_i\preceq u_i\preceq \overline u_i.\yesnumber
\end{IEEEeqnarray*}

It is straightforward to see that the sub-problem~\eqref{eq:opt_second_pro_2} is a constrained quadratic optimization problem with a lasso term, which can be solved by many existing approaches, such as semi-proximal ADMM. However, the computation efficiency is poor when solving the optimization problem~\eqref{eq:opt_second_pro_2} directly. As an alternative, we propose to solve the dual problem of~\eqref{eq:opt_second_pro_2}, because it can be formulated as a typical inequality constrained quadratic optimization problem, which can be solved apparently faster than the primal problem~\eqref{eq:opt_second_pro_2}.

For simplicity, define $\beta_i=z_i^k-\sigma^{-1}\lambda_i^k$, $\xi_i=(x_i,u_i)$, and the following notations:
\begin{IEEEeqnarray*}{rClCl}
\mathcal Q_i&=&\operatorname{diag}(Q_i,R_i)+\dfrac{\sigma}{2}I_{(m+n)T}&\in&\mathbb S^{(m+n)T}\\
q_i &=& -\Big(2Q_ir_i,0_{mT}\Big)-\sigma\beta_i&\in&\mathbb R^{(m+n)T}\\
\mathcal P&=&[0_{mT\times nT}\quad I_{mT}]&\in&\mathbb R^{mT\times (m+n)T}\\
\mathcal M_i&=&[-I_{nT},H_i]&\in&\mathbb R^{nT\times (n+m)T}\\
m_i&=&-G_ix_{i}^t&\in&\mathbb R^{nT}\\
\underline \xi_i&=&(\underline x_i,\underline u_i)&\in&\mathbb R^{(m+n)T}\\
\overline\xi_i&=&(\overline x_i,\overline u_i)&\in&\mathbb R^{(m+n)T}\\
\mathcal N_i &=& \big(I_{(m+n)T},-I_{(m+n)T}\big)&\in&\mathbb R^{2(m+n)T\times (m+n)T}\\
n_i &=& \Big(\overline \xi_i,-\underline  \xi_i\Big)&\in&\mathbb R^{2(m+n)T}.\yesnumber
\end{IEEEeqnarray*}
Then, the problem~\eqref{eq:opt_second_pro_2} is converted to
\begin{IEEEeqnarray*}{rl}\label{eq:primal_problem}
\min_{\xi_i,p_i} \quad
&\xi_i^\top \mathcal Q_i\xi_i+q_i^\top \xi_i +\alpha_i\|p_i\|_1 \\
\operatorname{subject\ to} \quad
&\mathcal M_i\xi_i=m_i\\
&\mathcal N_i\xi_i\preceq n_i\\
&\mathcal P\xi_i=p_i. \yesnumber
\end{IEEEeqnarray*}
Notably, the vectors $m_i$ and $n_i$ are constant vectors but the vector $p_i$ is a variable in the optimization problem.

\begin{assumption}\label{assumption:strong}
	Problem~\eqref{eq:primal_problem} is strictly feasible.
\end{assumption}
Under Assumption~\ref{assumption:strong}, the Slater’s condition is satisfied. Therefore, strong duality always holds for the proposed optimization problem, and the optimal solution to the linear quadratic optimization problem~\eqref{eq:primal_problem} can be obtained by solving the corresponding dual problem, due to the difficulty to deal with the primal problem~\eqref{eq:opt_second_pro_2} directly.

Then, it follows that the dual problem of the proposed optimization problem can be derived. To find the dual problem, the Lagrangian function with respect to the sub-problem is denoted by
\begin{IEEEeqnarray*}{l}
\mathcal L_i(\xi_i,p_i;\mu_{i1},\mu_{i2},\mu_{i3})=\xi_i^\top \mathcal Q_i\xi_i+q_i^\top \xi_i +\alpha_i\|p_i\|_1 \\
+\big\langle\mu_{i1},\mathcal M_i\xi_i-m_i\big\rangle+\big\langle\mu_{i2},\mathcal N_i\xi_i- n_i\big\rangle+\big\langle\mu_{i3},\mathcal P\xi_i-p_i\big\rangle,\\\yesnumber
\end{IEEEeqnarray*}
where $\mu_{i1}\in\mathbb R^{nT}$, $\mu_{i2}\in\mathbb R^{2(m+n)T}$, and $\mu_{i3}\in\mathbb R^{mT}$ are the Lagrangian dual variables. Then, the dual problem can be denoted by
\begin{IEEEeqnarray*}{rl}\label{eq:dual_problem}
 \max_{\mu_{i1},\mu_{i2},\mu_{i3}} \quad 
& \min_{\xi_i,p_i}\quad \mathcal L_i(\xi_i,p_i;\mu_{i1},\mu_{i2},\mu_{i3}) \\
\operatorname{subject\ to} \quad
&\mu_{i2}\succeq 0.\yesnumber
\end{IEEEeqnarray*}

In the following, the parameter of the Lagrangian function is ignored for the simplicity of the notation. Then, the minimization of the Lagrangian function $\mathcal L_i$ can be separated into two parts: the minimization with respect to $\xi_i$ and the minimization with respect to $p_i$:
\begin{IEEEeqnarray*}{rl}\label{eq:min_mathcal_L}
\min_{\xi_i,p_i} \mathcal L_i &= \min_{\xi_i}\Big\{ \xi_i^\top \mathcal Q_i\xi_i+q_i^\top \xi_i\\
&+\big\langle\mu_{i1},\mathcal M_i\xi_i\big\rangle+\big\langle\mu_{i2},\mathcal N_i\xi_i\big\rangle+\big\langle\mu_{i3},\mathcal P\xi_i\big\rangle\Big\}\\
&+\min_{p_i}\Big\{\alpha_i\|p_i\|_1-\big\langle\mu_{i3},p_i\big\rangle\Big\}-\big\langle\mu_{i1},m_i\big\rangle-\big\langle\mu_{i2}, n_i\big\rangle.\\\yesnumber
\end{IEEEeqnarray*}
Since the first part in~\eqref{eq:min_mathcal_L} is an unconstrained optimization problem, thus its optimum can be determined by setting the gradient to be zero, that is
\begin{IEEEeqnarray}{rCl}
2\mathcal Q_i\hat \xi_i+q_i+\mathcal M_i^\top \mu_{i1}+\mathcal N_i^\top \mu_{i2}+\mathcal P^\top \mu_{i3}=0,
\end{IEEEeqnarray}
where $\hat \xi_i$ denotes the optimal $\xi_i$ in the optimization problem. Therefore, we have
\begin{IEEEeqnarray}{rCl}
\hat \xi_i =-\frac{1}{2}\mathcal Q_i^{-1}\Big(q_i+\mathcal M_i^\top \mu_{i1}+\mathcal N_i^\top \mu_{i2}+\mathcal P^\top \mu_{i3}\Big).
\end{IEEEeqnarray}
Notably, because the matrices $Q_i$ and $R_i$ are positive semi-definite, the matrix $\mathcal Q_i$ is non-singular.

To find the minimization solution of the second part in~\eqref{eq:min_mathcal_L}, we introduce the definition of the conjugate function. 
\begin{definition}
The conjugate function of a function $f$ is defined as
\begin{IEEEeqnarray}{rCl}
f^*(y)&=&\sup\{\langle x,y\rangle-f(x)\}.
\end{IEEEeqnarray}
\end{definition}
Define $h(p_i) = \alpha_i\|p_i\|_1$. Then, it is straightforward to see that the second part in~\eqref{eq:min_mathcal_L} is the conjugate function of $h$, which can be denoted by $h^*(\mu_{i3})$. It is well-known the conjugate function of a norm function is an indicator function with respect to the corresponding dual norm ball, that is 
\begin{IEEEeqnarray}{rl}
h^*(\mu_{i3})=\left\{
\begin{array}{rl}
0,& \text{if }\|\mu_{i3}\|_\infty \le \alpha_i\\
\infty, & \text{otherwise}.
\end{array}\right.
\end{IEEEeqnarray}
Then, we have 
\begin{IEEEeqnarray*}{l}
\min_{\xi_i,p_i}\quad \mathcal L_i =-\frac{1}{4}\Big(q_i+\mathcal M_i^\top \mu_{i1}+\mathcal N_i^\top \mu_{i2}+\mathcal P^\top \mu_{i3}\Big)^\top \mathcal Q_i^{-1}\\
\qquad\quad\quad\Big(q_i+\mathcal M_i^\top \mu_{i1}+\mathcal N_i^\top \mu_{i2}+\mathcal P^\top \mu_{i3}\Big)-h^*(\mu_{i3}).\IEEEeqnarraynumspace\yesnumber
\end{IEEEeqnarray*}
Subsequently, the dual problem~\eqref{eq:dual_problem} is given by
\begin{IEEEeqnarray*}{rl}
\displaystyle \min_{\mu_{i1},\mu_{i2},\mu_{i3}}\quad &
\dfrac{1}{4}\Big(q_i+\mathcal M_i^\top \mu_{i1}+\mathcal N_i^\top \mu_{i2}+\mathcal P^\top \mu_{i3}\Big)^\top \mathcal Q_i^{-1}\\
&\Big(q_i+\mathcal M_i^\top \mu_{i1}+\mathcal N_i^\top \mu_{i2}+\mathcal P^\top \mu_{i3}\Big)\\
&+\big\langle\mu_{i1},m_i\big\rangle+\big\langle\mu_{i2}, n_i\big\rangle \\
\operatorname{subject\ to} \quad
&\mu_{i2}\succeq 0\\
&\|\mu_{i3}\|_\infty \le \alpha_i.\yesnumber
\end{IEEEeqnarray*}
Define the following vectors and matrices:
\begin{IEEEeqnarray*}{l}
\mu_i = (\mu_{i1},\mu_{i2},\mu_{i3})\in\mathbb R^{3(m+n)T}\\
\Omega_i = \Big(\mathcal M_i,\mathcal N_i,\mathcal P\Big)\in\mathbb R^{3(m+n)T\times (m+n)T}\\
\omega_i = \Big(m_i,n_i,0_{mT}\Big)\in\mathbb R^{3(m+n)T}\\
\Phi_i=\Big\{\mu_i = (\mu_{i1},\mu_{i2},\mu_{i3})\, \Big|\, \mu_{i2}\succeq 0, \|\mu_{i3}\|_\infty \le \alpha_i\Big\}.\yesnumber
\end{IEEEeqnarray*}
It follows that the optimization problem can be denoted in a compact form:
\begin{IEEEeqnarray*}{rl}\label{eq:dual_prob1}
\displaystyle \min_{\mu_i} \quad &
\dfrac{1}{4}\Big(\Omega^\top _i\mu_i+q_i\Big)^\top \mathcal Q_i^{-1}\Big(\Omega^\top _i\mu_i+q_i\Big)+\big\langle \omega_i,\mu_i\big\rangle\\
\operatorname{subject\ to} \quad
&\mu_i\in\Phi.\yesnumber
\end{IEEEeqnarray*}

It is obvious that the optimization problem~\eqref{eq:dual_prob1} is an inequality constrained quadratic optimization problem, which has been widely explored and can be solved much more efficiently than the primal optimization problem~\eqref{eq:primal_problem}.

%
%
After the solution to the dual problem is obtained, the optimal $\xi$ can be calculated by
\begin{IEEEeqnarray*}{rCl}
	\xi_i^* =-\frac{1}{2}\mathcal Q_i^{-1}\Big(\Omega_i^\top \mu_{i}^*+q_i\Big),\yesnumber
\end{IEEEeqnarray*}
where $\mu_{i}^*$ denotes the value of the corresponding vector when the stopping criterion is reached.

After that, each agent sends the latest optimal value $x_i^{k+1}$ and $u_i^{k+1}$ to the other agents, and in each agent, the following updating rule is performed:
\begin{IEEEeqnarray}{rl}
c^{k+1}=z_c^k-\sigma^{-1}\lambda_c^k.
\end{IEEEeqnarray}

\subsection{Second Sub-Problem in ADMM}

The second sub-problem~\eqref{eq:ADMM_iterations_2} in the main ADMM iterations can be represented by
\begin{IEEEeqnarray*}{rl}
\displaystyle \min_z \quad 
& \delta_\mathcal C(z_c)+\dfrac{\sigma}{2}\Big\|z-y^{k+1}-\sigma^{-1}\lambda^k\Big\|^2\\
\operatorname{subject\ to} \quad
&E_{\tau \ell}(z_{xi}-z_{xj})+\varepsilon z_{c,ij\tau \ell}-d\ge0\\
&5-\displaystyle\sum_{\ell=1}^6z_{c,ij\tau \ell}\ge0\\
&\forall i=1,2,\dotsm,N\\
&\forall j= 2,3,\dotsm,N,j>i\\
&\forall \tau=t,t+1,\dotsm,t+T-1\\
&\forall \ell=1,2,\dotsm,6,\yesnumber
\end{IEEEeqnarray*}
where the location of $z_{c, ij\tau \ell}$ in the vector $z$ is the same as the location of $c_{ij\tau \ell}$ in the vector $y$. 

Define the new linear operator $\mathcal E_{ij\tau\ell}: \mathbb R^{[m+n+3N-3]TN}\rightarrow\mathbb R$ as 
\begin{IEEEeqnarray*}{rCl}
\mathcal E_{ij\tau0}(z)&=&-\displaystyle\sum_{\ell=1}^6z_{c,ij\tau \ell}\\
\mathcal E_{ij\tau\ell}(z)&=&E_{\tau \ell}(z_{xi}-z_{xj})+\varepsilon z_{c,ij\tau \ell}\\
\forall i&=&1,2,\dotsm,N\\
\forall j&=&2,3,\dotsm,N, j>i\\
\forall \tau &=& t,t+1,\dotsm,t+T-1\\
\forall \ell&=&1,2,\dotsm,6, \yesnumber
\end{IEEEeqnarray*}
and the linear operator $\mathcal E: \mathbb R^{[m+n+(13/2)N-(13/2)]TN}\rightarrow\mathbb R^{(7/2)N(N-1)T}$ that contains all the linear operator $\mathcal E_{ij\tau\ell}$ for all $i=1,2,\dotsm,N$, $j=2,3,\dotsm,N,j>i$, $\tau=t,t+1,\dotsm,t+T-1$, $\ell=0,1,\dotsm,6$ in the ascending order. Also, define the vector $\mathcal E_c \in\mathbb R^{(7/2)N(N-1)T}$ as
\begin{IEEEeqnarray*}{rl}
\mathcal E_c=\Big(\underbrace{5,\underbrace{-d,-d,\dotsm,-d}_6,\dotsm,
	5,\underbrace{-d,-d,\dotsm,-d}_6}_\text{There are $(1/2)N(N-1)T$ repetition}\Big). \yesnumber
\end{IEEEeqnarray*}
Further, define the constant $\gamma^k$ in the $k$th iteration as 
\begin{IEEEeqnarray}{l}
\gamma^k = y^{k+1}+\sigma^{-1}\lambda^k,
\end{IEEEeqnarray}
and then the optimization problem can be represented by
\begin{IEEEeqnarray*}{rl}\label{eq:ADMM2_compact}
\displaystyle \min_z \quad
&\Big\|z-\gamma^k\Big\|^2+\delta_\mathcal C(z_c)\\
\operatorname{subject\ to} \quad
&\mathcal E(z)+\mathcal E_c\succeq 0.\yesnumber
\end{IEEEeqnarray*}
Remarkably, \eqref{eq:ADMM2_compact} is a constrained least square optimization problem with integer optimization variables, and feasible solutions can be readily obtained by many existing methods. 

\subsection{Initialization of Optimization Variables}\label{section:init}

Since the problem~\eqref{eq:outer_admm_prob} is a non-convex optimization problem, it suffers from NP-hardness to find the global optimum. However, some intuitive initialization procedures are straightforward to be implemented to accelerate the optimization process. To introduce the initialization procedures, the following matrices are introduced:
\begin{IEEEeqnarray*}{rCl}
E_{ij\tau\ell}(z)&=&E_{\tau \ell}(z_{xi}-z_{xj})\\
\forall i&=&1,2,\dotsm,N\\
\forall j&=&2,3,\dotsm,N, j>i\\
\forall \tau &=& t,t+1,\dotsm,t+T-1\\
\forall \ell&=&1,2,\dotsm,6.\yesnumber
\end{IEEEeqnarray*}
Then, define the linear operator $E: \mathbb R^{[m+n+3N-3]TN}\rightarrow\mathbb R^{3N(N-1)T}$, which contains all the linear operator $ E_{ij\tau\ell}$ for all $i=1,2,\dotsm,N$, $j=2,3,\dotsm,N,j>i$, $\tau=t,t+1,\dotsm,t+T-1$, $\ell=1,\dotsm,6$ in the ascending order. Notably, the linear operator $E$ gives the information of the distances between different agents. 

In order to determine the initial optimization variables, we can use the linear operator $E$ to find a feasible vector $z_c$. To be more specific, if the distance between a pair of agents in a direction is greater or equal than the safe distance $d$, then the corresponding entry in the vector $z_c$ will be $0$, and otherwise it will be $1$. With this technique, the initial parameter can be determined easily.

\begin{remark}
	Based on the proposed initialization procedures, the optimum can be reached for the sub-problem~\eqref{eq:ADMM_iterations_2} directly, under the condition that each pair of agents maintains a safe distance. This is because the vector $z_c$ is chosen  as a feasible solution with the proposed initialization procedures. It is also pertinent to note that if a safe distance is not maintained initially, the proposed initialization procedures will fail because the initial choice of the vector $z_c$ will not be feasible.
\end{remark}

\begin{remark}
	The control law of the proposed model predictive controller is the optimal solution to the optimization problem~\eqref{eq:outer_admm_prob} in a receding time window, and the controller parameters are always changed at each time stamp. To be more specific, a receding horizon control (RHC) structure based on the solution of the MPC problem is utilized in the proposed algorithm.
\end{remark}

\begin{remark}
	In terms of the stability of the closed-loop system, there are many existing works that report the results on the stability of the typical model predictive controller such as~\cite{9109670, 8637808, eskandarpour2020constrained}.
\end{remark}

\begin{remark}
	It is well-known that increasing the prediction horizon can enlarge the domain of attraction of the MPC controller. However, this will increase the computational complexity. Weighting the terminal cost can also enlarge the domain of attraction of the MPC controller without the occurrence of the terminal constraints; thus, the stabilizing weighting factor of a given initial state can be included, which has been proved in~\cite{limon2006stability} that the asymptotic stability of the nonlinear MPC controller can be achieved.
\end{remark}

The proposed parallel optimization method is summarized as Algorithm~\ref{algo:1}.

\begin{algorithm}
	\caption{ADMM for Parallel MPC Problem}
	\begin{algorithmic}[1]\label{algo:1}
		\renewcommand{\algorithmicrequire}{\textbf{Input:}}
		\renewcommand{\algorithmicensure}{\textbf{Output:}}
		\REQUIRE weighting matrices for each agent $Q_i$ and $R_i$; box constraint parameter vectors for each agent $\underline x_i, \overline x_i, \underline u_i, \overline u_i$; reference signal for each agent $r_i$;  initial system state vector for each agent $x^0_i$; penalty parameter $\sigma$ in the augmented Lagrangian function; prediction horizon $T$; number of total time stamps $T_t$; maximum ADMM iteration number $K$; stopping criterion $\hat \epsilon$.
		\FOR{$t=0,1,\dotsm,T_t$}
		\STATE Choose the initial optimization variable $y^0$.
		\FOR{$k=0,1,\dotsm,K$}
		\STATE Update the optimization variable $y^{k+1}$ by~(\ref{eq:ADMM_iterations}a) on each agent \textbf{parallelly} and share the results to the central controller.
		\STATE Update the consensus variable $z^{k+1}$ by~(\ref{eq:ADMM_iterations}b) on the central controller. 
		\STATE Update the Lagrangian multiplier $\lambda^{k+1}$ by~(\ref{eq:ADMM_iterations}c) on the central controller.
		\STATE Calculate $\epsilon_{pri}$ on the central controller.
		\IF {$\epsilon_{pri}\le \hat \epsilon$}
		\STATE Update the optimal control input vector for each agent $u^*_{i}$ from the vector $z^{k+1}$.
		\STATE \textbf{break}
		\ENDIF
		\ENDFOR
		\STATE Apply the control input signal at the first time stamp in $u^*_{i}$ to the system dynamic model for each agent.
		\STATE Calculate the initial state vector for each agent $x^{t+1}_i$.
		\ENDFOR
		\RETURN feasible trajectory for each agent $x^0_i,x^1_i,\dots,x^{T_t}_i$.
	\end{algorithmic} 
\end{algorithm}
\section{Convergence Analysis}
In this section, the proof of the convergence with the proposed ADMM iterations is given. Also, a new ADMM iterations scheme is applicable to solve the proposed optimization problem. Analysis and relationship between the two schemes are discussed after the convergence proof.
\subsection{Convergence Analysis for the Proposed ADMM Iterations}
To analyze the convergence of the optimization problem~\eqref{eq:outer_admm_prob} and simplify the notations, we define the objective function of the optimization problem~\eqref{eq:outer_admm_prob} as $\phi(y,z)$ and then separate $\phi(y,z)$ into two parts $\varphi(y)$, $\psi(z)$, where
\begin{IEEEeqnarray*}{rCl}
	\phi(y,z) &=& \mathbf 1^\top f(y)+\delta_{\mathcal Y_1}(y)+\mathbf 1^\top  g(z)+\delta_{\mathcal Y_2}(z)\\
	\varphi(y)&=&\mathbf 1^\top f(y)+\delta_{\mathcal Y_1}(y)\\
	\psi(z) &=& \mathbf 1^\top g(z)+\delta_{\mathcal Y_2}(z) \yesnumber
\end{IEEEeqnarray*}

Before we analyze the convergence of the proposed non-convex ADMM, the following definition of the coercive function is given.
\begin{definition}
	If a function $\phi(y,z)$ is coercive over a set $\mathscr F$, then $\phi(y,z) \rightarrow \infty$ if $(y,z)\in\mathscr F$ and $\|(y,z)\|\rightarrow \infty$.
\end{definition}
Then, the following lemma (which we had previously established) 
indicates
that the objective function $\phi(y,z)$ is coercive.
Because of constraints of space,  
the proofs of Lemma 1 below, and also subsequently of Lemmas 2,3,4;
previously already rigorously established by us in the earlier version of the manuscript,
are moved
to an independent supplementary file.
\begin{lemma}\label{lemma:coecivty}
	Define a feasible set $\mathscr F = \{(y,z)\in\mathbb R^{(m+n+3N-3)TN}\times \mathbb R^{(m+n+3N-3)TN}\,|\, y-z=0\}$. In the feasible set $\mathscr F$, the objective function $\phi(y,z)$ is coercive.
\end{lemma}
Next, the following lemma (likewise previously established by us)
indicates
the property of the Lipschitz sub-minimization paths in terms of the optimization problem~\eqref{eq:outer_admm_prob}.
\begin{lemma}\label{lemma:Lip_sub_min_path}
	The following statements hold for~\eqref{eq:outer_admm_prob_converted}:
	\begin{enumerate}[(a)]
		\item For any fixed $y$, the problem $\operatorname{argmin}_z\{\phi(y,z)\,|\, z=u \}$ has a unique minimizer, and $u\mapsto \operatorname{argmin}_z\{\phi(y,z)\,|\, z=u \}$ is a Lipschitz continuous map.
		\item For any fixed $z$, the problem $\operatorname{argmin}_y\{\phi(y,z)\,|\, y=u \}$ has a unique minimizer, and $u\mapsto \operatorname{argmin}_y\{\phi(y,z)\,|\, y=u \}$ is a Lipschitz continuous map.
		\item The two maps in Statement (a) and (b) have a universal Lipschitz constant $\bar M$, and $\bar M>0$.
	\end{enumerate}
\end{lemma}


To indicate a pertinent property of the function $\psi(z)$, the definition of semi-continuous function is stated;
and a lemma result on it (again previously established by us) is listed.
\begin{definition}
	A function $\phi$ is semi-continous at $x_0$ if for all $y$ such that $y<\phi(x_0)$, there exists a neighbourhood $\mathcal X$ of $x_0$ such that $y<\phi(x)$ for all $x$ in $\mathcal X$. A function $\phi$ is semi-continuous if $\phi$ is semi-continuous at every point in the domain of $\phi$.
\end{definition}
\begin{lemma}\label{lemma:semi-continuous}
	The function $\psi(z)$ is lower semi-continuous.
\end{lemma}

Next, the definition of the restricted prox-regularity is firstly introduced to generalize the definition of the Lipschitz differentiable function. Then, the definition of the Lipschitz differentiable function is presented and generalized to be applied to the proposed optimization problem. Since the indicator function is included in the function $\varphi(z)$, it leads to the difficultly in analyzing the properties of the function $\varphi(z)$.
\begin{definition}\label{def:restricted_prox_regular}
	Define the lower semi-continuous function $\varphi(z):\mathbb R^{(m+n+3N-3)TN}\rightarrow \mathbb R\cup\{\infty\}$, $M>0$, and define the set $\mathscr H =\{z\in\operatorname{dom}(\varphi)\,|\, \|d\|>M, \forall d\in\partial \varphi(z)\}$. Then, $\varphi(z)$ is restricted prox-regular if for any $M>0$, and bounded set $T\subseteq \operatorname{dom}(\varphi)$, there exists $\varpi>0$ such that $\phi(z)+\frac{\varpi}{2}\|x-z\|\ge \varphi(x)+\langle d,z-x\rangle, \forall z\in T\backslash \mathscr H,x\in T,d\in\partial \varphi(z),\|d\|\le M$.
\end{definition}

\begin{definition}\label{def:Lipschitz}
	A function $\varphi(z)$ is Lipschitz differentiable if the function $\varphi(z)$ is differentiable and the gradient is Lipschitz continuous.
\end{definition}

Notably thence, a function $\varphi(z)$ is Lipschitz differentiable if the function $\varphi(z)$ is restricted prox-regular;
and the following lemma result (also too, previously already established by us) is thus appropriate to state.

\begin{lemma}\label{lemma:differentiable}
	The function $\varphi(z)$ is Lipschitz differentiable with the constant $\mathscr L_\varphi$.
\end{lemma}

With these, the following main result can be stated:

\begin{theorem}\label{theorem: convergence}
	The ADMM iterations~\eqref{eq:ADMM_iterations} converge to the limit point for any sufficiently large $\sigma$. Equivalently, given any starting point $(y^0, z^0, \lambda^0)$, the proposed ADMM iterations~\eqref{eq:ADMM_iterations} provide a bounded sequence, which converges to the limit point $(y^*,z^*,\lambda^*)$. The converged limit point $(y^*,z^*,\lambda^*)$ is a stationary point of the augmented Lagrangian function $\mathcal L_\sigma(y,z;\lambda)$, that is $0\in\partial \mathcal L_\sigma(y,z;\lambda)$.
\end{theorem}
\begin{proof}
	From Lemmas 1-4 and the fact that the equality constraint $y=z$ is feasible, the convergence assumptions required in~\cite{wang2019global} are satisfied. This proves Theorem~\ref{theorem: convergence}.
\end{proof}

\begin{remark}
	From Theorem~\ref{theorem: convergence}, it is evident that a sufficiently large penalty parameter in the augmented Lagrangian function  guarantees the local convergence. This notwithstanding, if $\sigma$ is chosen to be overly large, there will be a trade-off in the convergence rate of the proposed algorithm.
Additionally, even though the proof of convergence requires the objective function to be coercive, 
in practice,
it is observed
that convergence is also attained (in simulation experiments) 
when the cost function is just simply positive semi-definite.
\end{remark}



\subsection{Discussion on the Convergence of ADMM Iterations in Other Forms}
In this subsection, we discuss on the relationship between different forms of the ADMM iterations. 
Ahead of this discussion, 
an appropriate pertinent 
main 
result,
Theorem~\ref{theorem: equivalent_problem}, 
is first developed and proved here
to establish and set in place
the equivalence of the two optimization problems.
\begin{theorem}\label{theorem: equivalent_problem}
	The optimization problem~\eqref{eq:outer_admm_prob} is equivalent to the following optimization problem:
	\begin{IEEEeqnarray*}{rl}\label{eq:outer_admm_prob_converted}
		\min_{y,z} \quad
		& \mathbf 1^\top f(y)+\mathbf 1^\top g(z)+\delta_{\mathcal Y_2}(z)\\
		\operatorname{subject\ to} \quad
		&\mathscr Hy-\mathscr Gz=\mathscr F,\yesnumber
	\end{IEEEeqnarray*}
	where the matrices $\mathscr H\in\mathbb R^{(2n+m+3N-3)NT\times (n+m+3N-3)}$ and $\mathscr G\in\mathbb R^{(2n+m+3N-3)NT\times (n+m+3N-3)}$ are defined as
	\begin{IEEEeqnarray*}{rCl}
		\mathscr H &=& \bigg(\Big[\operatorname{diag}\big(\underbrace{[-I_{nT},H_1],\dotsm,[-I_{nT},H_N]}_{N}\big),\\
		&&0_{3N(N-1)T\times 3N(N-1)T}\Big], I_{(n+m+3N-3)NT}\bigg)\\
		\mathscr G &=& \Big(0_{nNT\times (n+m+3N-3)NT},I_{(n+m+3N-3)NT}\Big)\\
		\mathscr F &=& \Big(-G_1x_{1}^t,-G_2x_{2}^t,\dotsm,-G_Nx_{N}^t,0_{3(N-1)NT}\Big).\IEEEeqnarraynumspace\yesnumber
	\end{IEEEeqnarray*}
\end{theorem}
\begin{proof}
	For this, note that the indicator function $\delta_{\mathcal Y_1}(y)$ of dynamic constraints in the first part of the problem~\eqref{eq:outer_admm_prob} is combined into the matrix $\mathscr H$. 
The equivalent constraint $y=z$ is also included into the matrices $\mathscr H$ and $\mathscr G$ by introducing two identity matrices. 
This proves Theorem~\ref{theorem: equivalent_problem}.
\end{proof}

Define the augmented Lagrangian function ${\hat {\mathcal L}}_\sigma$ with respect to the optimization problem~\eqref{eq:outer_admm_prob_converted} as 
\begin{IEEEeqnarray*}{rl}\label{eq:augmented_lagrangian_function_2}
	{\hat {\mathcal L}}_\sigma\Big(y,z;\hat \lambda\Big)&=\mathbf 1^\top f(y)+\mathbf 1^\top g(z)+\delta_{\mathcal Y_2}(z)\\
	&+\dfrac{\sigma}{2}\|\mathscr Hy-\mathscr Gz-\mathscr F+\sigma^{-1}\hat \lambda\|^2-\dfrac{1}{2\sigma}\|\hat \lambda\|^2.\IEEEeqnarraynumspace\yesnumber
\end{IEEEeqnarray*}
where $\hat \lambda^{k} \in\mathbb R^{(2n+m+3N-3)NT}$ denotes the Lagrangian multiplier. Then, the ADMM iterations in terms of the augmented Lagrangian function~\eqref{eq:augmented_lagrangian_function_2} are given by
\begin{IEEEeqnarray*}{rCl}\label{eq:ADMM_iterations2}
	y^{k+1}&=&\displaystyle\operatorname*{argmin}_y \;\hat {\mathcal L}_\sigma\Big(y,z^k;\hat \lambda^k\Big)\\
	z^{k+1}&=&\displaystyle\operatorname*{argmin}_z \;\hat {\mathcal L}_\sigma\Big(y^{k+1},z;\hat \lambda^k\Big)\\
	\hat \lambda^{k+1}&=& \hat \lambda^k+\sigma\Big(\mathscr Hy^{k+1}-\mathscr Gz^{k+1}-\mathscr F\Big).\yesnumber
\end{IEEEeqnarray*}

Theorem~\ref{theorem: equivalent_problem} shows that the two optimization problems are equivalent.
Also too, the optimization variables in the iterations~\eqref{eq:ADMM_iterations} 
are the approximate solutions to the optimization variables in the iterations~\eqref{eq:ADMM_iterations2} 
by transferring the dynamic and box constraints into the objective function.
The corresponding ADMM iterations are consequently equivalent. 

It is obvious that the ADMM iterations~\eqref{eq:ADMM_iterations2} are also straightforward to apply, and the convergence rate should also be the first-order convergence due to the nature of the ADMM algorithm. However, it is noteworthy to mention that the ADMM iterations in our proposed method have less optimization variables and thus it can save the memory during the algorithm implementation. Experimental results also demonstrate faster convergence with the former ADMM iterations. Therefore, the ADMM iterations of the problem~\eqref{eq:outer_admm_prob} are applied in this paper.

\begin{remark}
For our proposed algorithm, the computational and space complexity with respect to the number of agents $N$ is of the order of $O(N^2)$. 
This arises from our algorithm
where
a distance matrix between each pair of agents 
is straightforwardly required.
With this rather straightforward and direct order of $O(N^2)$ complexity attained in our approach,
it can be seen
that
this will \emph{not} be of 
the so-termed 
``super-computing/high-performance computing (HPC)''
requirement 
(which would otherwise have been the case if the computational complexity
is instead of an exponential order, for example).
The development thus of the
approach here,
with 
this rather straightforward and direct order of $O(N^2)$ complexity,
is actually a certainly notable outcome.
This outcome thus means that
indeed,
the approach developed here
certainly attains the important requirement of
practical scalability and generalizability.
\end{remark}
For now,
at this stage of 
our experimental and development efforts,
and specifically with the constraints of
our own
current practical available computational resources, 
indeed in experiments,
we can currently only deal with multi-agent systems
with more reasonably sized, although certainly not actually small, dimensions 
(up to about $N=21$ agents)
at this stage. 
In so far as the key theoretical properties of the methodology are concerned thence, 
these had been rigorously established in
Theorems 1 and 2. 
There thus indeed nevertheless remains further pertinent
substantial
future project
investigations and developments
on the experimental and computational aspects
(which will certainly be of the ilk of a further large and significant project on its own),
that
will certainly suitably complement
the results and outcomes in the
theoretical and 
algorithmic development 
aspects of the proposed control method
that have been attained here;
on further practical experiments 
with a very large number of agents.
This will definitely be embarked 
upon as a future substantial project.


\section{Illustrative Example}

\subsection{Simulation Setup} \label{section.simulation_setup}
In this section, a multi-agent formation control task for a group of UAV is implemented as an illustrative example, which is introduced to present the performance of the proposed parallel model predictive collision-free control method. In this task, a multi-UAV system with 21 agents ($N=21$) is utilized, and these agents are required to track the given reference considering all the given constraints. For the reference profile, 
there are four key points for each agent to track. The initial points, which are different for each agent, are chosen in a square with an equal interval of $1$ m in the $yz$-plane and $-2$ on the $x$-axis. The second point and the third point, which are the same for all the agents, are chosen as $(-1,-1,-1)$ and $(1,1,1)$, respectively. If there is no collision avoidance constraint, all the agents will pass the points $(-1,-1,-1)$ and $(1,1,1)$ along the reference. The final points are chosen as the points with the same position as the initial points in the $yz$-plane and $2$ on the $x$-axis. Notably, all agents depart from the initial point and arrive at the final point at the same time. This reference profile design means that the system state and control input variables must be restricted in a feasible range and the safe distance must be guaranteed for each pair of agents. Also notably here, to realize a tracking optimization problem and avoid huge cost function values during the tracking and control the speed accurately for each UAV, we connect all the set points with a time-varying reference signal for each UAV. Specifically, the reference points in the reference signal are evenly distributed, and thus a trajectory tracking MPC problem is finally formulated and solved. 


Before the simulation results are shown, we present the modeling of a UAV system.

\subsection{UAV Modeling}
To clearly illustrate the UAV modeling, the notations in this section are not related to the notations used in the previous sections. Assume that aerodynamic and gyroscopic effects are not taken into consideration in the dynamic model. We use the dynamic model of the quadcopter in \cite{8971490} with the same parameter settings:
\begin{IEEEeqnarray*}{rCl}\label{dyn}
	\dot{p}(t) &=&R(\phi,\theta,\psi)^\top  {v}(t) \\ 
	\dot{v}(t) &=&-\langle \omega(t),{v}(t)\rangle +gR(\phi,\theta,\psi){e} + \frac{1}{m}{e}T_t\\ 
	\dot{\zeta}(t) &=&W(\phi,\theta,\psi) {\omega}(t) \\ 
	\dot{\omega}(t) &=&J^{-1}\Big(\langle -{\omega}(t), J{\omega}(t)\rangle + {\tau}\Big),\IEEEyesnumber
\end{IEEEeqnarray*}
where $p=(p_x, p_y, p_z)$ is the state vector representing the position; $R(\phi,\theta,\psi)$ and $W(\phi,\theta,\psi)$ are the rotation matrices and angular Jacobian, respectively; ${v}=(v_x, v_y, v_z)$ is the state vector representing the velocity; $\omega=(\omega_x, \omega_y, \omega_z)$ is the state vector representing the angular velocity; $g$ is the acceleration of gravity; $m$ denotes the helicopter mass; $e=(0,0,1)$ is the standard basis vector; $T_t$ denotes the total thrust; $\zeta=(\phi,\theta,\psi)$ is the angle state vector; ${\tau}=(\tau_x,\tau_y,\tau_z)$ represents the torques of the quadcopter in each dimension;  $J=\textup{diag}(J_x, J_y, J_z)$ is the moment of the inertia of the quadcopter. Notice that the rotation matrix and angular Jacobian are represented by
\begin{IEEEeqnarray*}{rCl}
	{R}(\phi, \theta, \psi)&=&
	\begin{bmatrix}
		{\mathbf c\theta \mathbf c\psi} & {\mathbf c\theta \mathbf s\psi} & {-\mathbf s\theta} \\ 
		{\mathbf s\theta \mathbf c\psi \mathbf s\phi-\mathbf s\psi \mathbf c\phi} & {\mathbf s\theta \mathbf s\psi \mathbf s\phi+\mathbf c\psi \mathbf c\phi} & {\mathbf c\theta \mathbf s\phi} \\
		{\mathbf s\theta \mathbf c\psi \mathbf c\phi+\mathbf s\psi \mathbf s\phi} & {\mathbf s\theta \mathbf s\psi \mathbf c\phi-\mathbf c\psi \mathbf s\phi} & {\mathbf c\theta \mathbf c\phi}
	\end{bmatrix}\\
	W(\phi, \theta, \psi)&=&\begin{bmatrix}{1} & {\mathbf s\phi \mathbf t\theta} & {\mathbf c\phi \mathbf t\theta} \\ {0} & {\mathbf c\phi} & {-\mathbf s\phi} \\ {0} & {\mathbf s\phi \mathbf{sc}\theta} & {\mathbf c\phi \mathbf{sc}\theta}\end{bmatrix},\IEEEyesnumber
\end{IEEEeqnarray*}
where the functions $\sin(\cdot),\cos(\cdot),\tan(\cdot),\sec(\cdot)$ are represented by $\mathbf c, \mathbf s,\mathbf t,\mathbf {sc}$ for the sake of simplicity.

Assume that the dynamics of the rotor are ignored. Then, the vertical forces are directly considered as the system inputs, and a moment perpendicular to the plane of the propeller rotation is necessary to be considered. We have the following equation:
\begin{IEEEeqnarray*}{rCl}
	\label{control_input}
	\left[\begin{array}{l}
		{T_t} \\ {\tau_{1}} \\ {\tau_{2}} \\ {\tau_{3}}\end{array}\right]=\left[\begin{array}{cccc}{-1} & {-1} & {-1} & {-1} \\ {0} & {-L} & {0} & {L} \\ {L} & {0} & {-L} & {0} \\ {-c} & {c} & {-c} & {c}\end{array}\right]\left[\begin{array}{c}{F_{1}} \\ {F_{2}} \\ {F_{3}} \\ {F_{4}}
	\end{array}\right], \IEEEeqnarraynumspace\IEEEyesnumber
\end{IEEEeqnarray*}
where $L$ is the distance between the rotor and the center of the quadcopter; $c$ is the ratio of the rotor angular momentum, and the center of the vertical force. To denote the dynamic model in a compact manner, the system state vector is denoted by $x=(p, v, \zeta, \omega)$ and the system input vector is denoted by ${u} = (F_1, F_2, F_3, F_4)$. To make sure that the stable point is around the original point, we define $u=u_{eq}+\delta u$, where $u_{eq} = \left( {mg}/{4},\ {mg}/{4},\ {mg}/{4},\  {mg}/{4} \right)$ is used to offset the gravity. Then, the dynamic model of a quadcopter can be represented by
\begin{IEEEeqnarray*}{rCl}
	\label{eq:model}
	\begin{bmatrix} {\dot{p}} \\ {\dot{v}} \\ {\dot{\zeta}} \\ {\dot{\omega}}\end{bmatrix} =
	\begin{bmatrix} {R(\phi,\theta,\psi)}^\top  {v} \\ -\langle {\omega}, v\rangle + g {R(\phi,\theta,\psi)} e \\ W(\phi,\theta,\psi) {\omega} \\ J^{-1}\big(-\langle{\omega} , J {\omega}\rangle\big) \end{bmatrix} + B {(u_{eq} + \delta u)}, \IEEEyesnumber
\end{IEEEeqnarray*}
where the matrix $B$ is given by
\begin{IEEEeqnarray*}{rCl}
B = - \begin{bmatrix} 
0_{5\times 4} \\ 
\begin{bmatrix} {\frac{1}{m}} & {\frac{1}{m}} & {\frac{1}{m}} & {\frac{1}{m}} \end{bmatrix} \\ 
0_{3\times 4} \\ 
\begin{bmatrix} {0} &  \frac{L}{J_{x}} & 0 &  -\frac{L}{J_{x}} \\  -\frac{L}{J_{y}} & {0} & \frac{L}{J_{y}} & {0} \\ \frac{c}{J_{z}} & -\frac{c}{J_{z}} & \frac{c}{J_{z}} & -\frac{c}{J_{z}} \end{bmatrix}
\end{bmatrix}.\IEEEyesnumber
\end{IEEEeqnarray*}

To guarantee the feasibility of the MPC problem, the box constraints on the system state and input variables are required to be considered. In the proposed model, the system input $\delta u$ is restricted in a specific range, which is
\begin{IEEEeqnarray}{rCl}
	\delta {u}_{\min} \leq \delta {u} \leq \delta {u}_{\max},
\end{IEEEeqnarray}
where $\delta {u}_{\min}\in\mathbb{R}^4$ and $\delta {u}_{\max}\in\mathbb{R}^4$ can be chosen with regard to the requirements. The box constraints on the system state variables are given by
\begin{IEEEeqnarray}{rCl}
	-\pi \leq \phi \leq \pi,\ -\pi \leq \psi \leq \pi, 
	-\frac{\pi}{2} \leq \theta \leq \frac{\pi}{2}.
\end{IEEEeqnarray}

With the aforementioned dynamic model and box constraints as well as the MIP-based collision avoidance constraints, the multi-agent MPC problem~\eqref{eq.mpc2} can be formulated.
\subsection{Simulation Implementation}
To solve the nonlinear MPC problem, several algorithms have been proposed, such as the iterative linear quadratic regulator~\cite{thomas2021receding} and differential dynamic programming~\cite{pavlov2021interior}. A more straightforward approach to implementing the nonlinear MPC controller is to perform a state-dependent coefficient (SDC) factorization~\cite{zhang2019hlt}. For each factorization method, a linear system representation can be eventually obtained regarding the current system state variables, and the obtained linear system model can be straightforwardly represented in the form of a state-space model. Subsequently, a linear MPC optimization problem is formulated in each time stamp, and the optimization problem can be effectively solved with the proposed algorithm. Since the system~\eqref{eq:model} is nonlinear, the state-dependent coefficient factorization method is used to deal with the nonlinear dynamics, which results in a linear time variant (LTV) state-space dynamic model. Since the RHC law is utilized, and the linear dynamic matrices $A$ and $B$ are dependent on the current state vector $x$, we can consider the system matrices to be constant during the prediction horizon.
\textcolor{black}{Additionally, all parameters pertinent to the system model, optimization problem, and the simulation setup are presented in Table~\ref{Table:parameter}.}
\begin{table}[htb!]
\centering
	\textcolor{black}{\caption{Simulation and Parameter Settings}}\label{Table:parameter}
\begin{tabular}{|l|l|}
\hline
\textbf{Description}                                                    & \textbf{Value}                                                                                                                                    \\ \hline
Simulation environment                                                  & Python 3.7                                                                                                                                        \\\hline
Processor                                                               & \begin{tabular}[c]{@{}l@{}}Intel(R) Xeon(R) CPU \\ E5-2695 v3 @ 2.30GHz\end{tabular}                                                              \\\hline
Prediction horizon                                                      & $T = 25$                                                                                                                                          \\\hline
Sampling time                                                           & $\Delta t = 0.05$ s                                                                                                                              \\\hline
Number of agents                                                        & $N = 21$                                                                                                                                          \\\hline
\begin{tabular}[c]{@{}l@{}}Lower bound of \\ control input\end{tabular} & \begin{tabular}[c]{@{}l@{}}$\delta {u}_{\min} = (-1.96\ \textup{N}, -1.96\ \textup{N},$ \\ \qquad \qquad \enspace $-1.96\ \textup{N}, -1.96\ \textup{N})$\end{tabular}     \\\hline
\begin{tabular}[c]{@{}l@{}}Upper bound of\\ control input\end{tabular}  & \begin{tabular}[c]{@{}l@{}}$\delta {u}_{\max} = (1.96\ \textup{N}, 1.96\ \textup{N},$ \\ \qquad \qquad \enspace $1.96\ \textup{N}, 1.96\ \textup{N})$\end{tabular}         \\\hline
Safety distance                                                         &  $d=0.2$ m\\\hline
Stopping criterion                                                      & $\hat \epsilon = 0.1$                                                                                                                             \\\hline
Weighting matrices                                                      & \begin{tabular}[c]{@{}l@{}}$Q_i = \operatorname{diag}(I_3,0_{9\times9})$; \\ $R_i = 0_{4\times4}$\end{tabular}                                    \\ \hline
\end{tabular}
\end{table}

\subsection{Simulation Result}
\begin{figure}[t]
	\centering
	\includegraphics[trim=100 100 100 100, clip, width=0.85\linewidth]{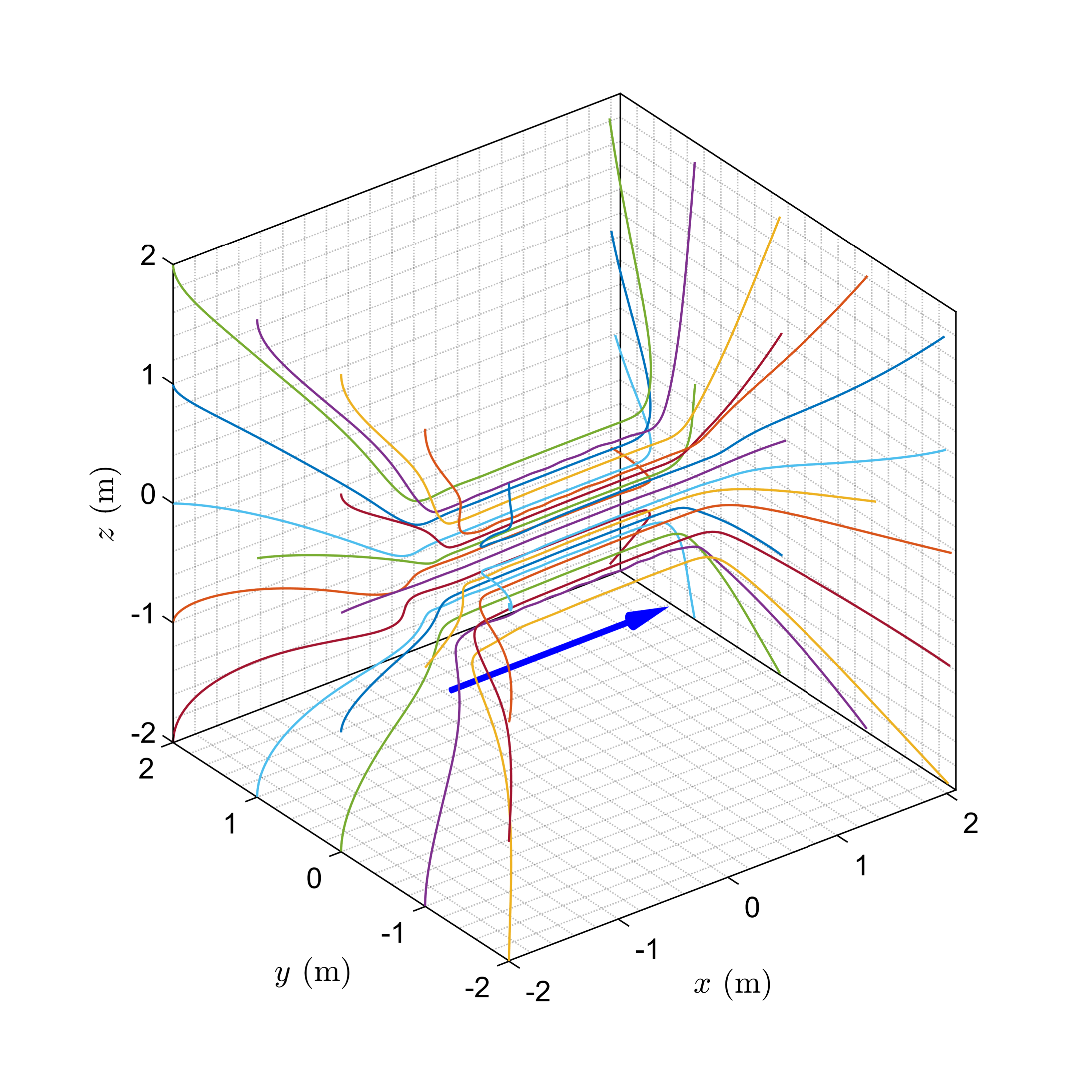}
	\caption{Trajectories of the agents in the 3D view.}
	\label{fig:Trajectory_3D}
\end{figure}
\begin{figure}[t]
	\centering
	\includegraphics[trim=100 20 100 50, clip, width=1\linewidth]{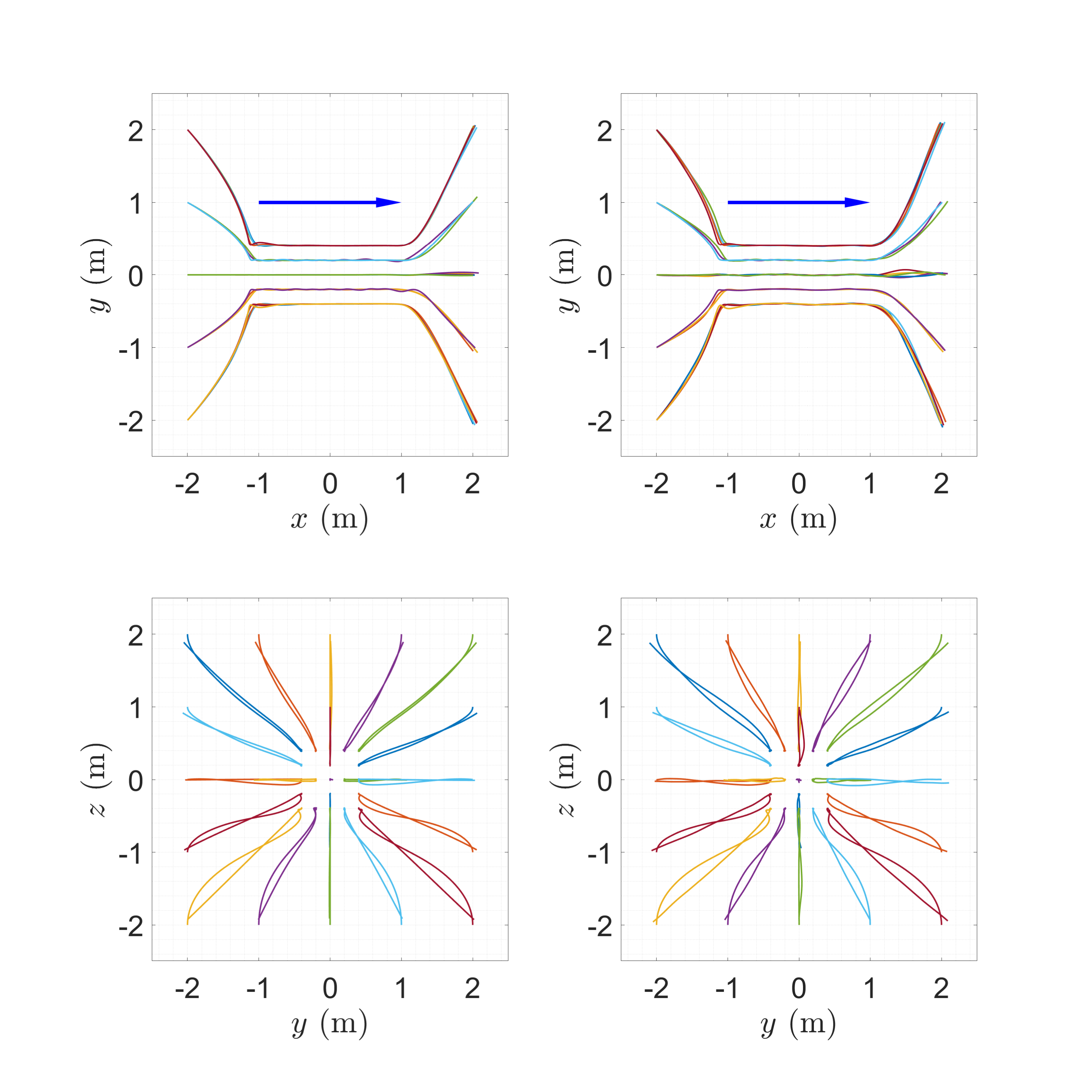}
	\caption{Trajectories of the agents in the 2D view without disturbance (left) and with disturbance (right).}
	\label{fig:Trajectory_2D}
\end{figure}

Before analyzing the simulation results, it is also worth noting the performance metrics in this work. Firstly, the trajectories of all agents are observed. As mentioned in Section~\ref{section.simulation_setup}, each agent has different initial and final set points, while all agents have same second and third set points. We can intuitively imagine that, it would be better if all agents are moving as close as possible to the set points at a safe distance. Therefore, to judge the performance of the proposed method, the trajectories of all agents are required to be observed. Secondly, as designed in this work, the box constraints of the control inputs of all agents are required to be satisfied. This is also an important performance metric. Thirdly, since the safe distance is set as $0.2$ m in this work, the distance between each pair of agents should be larger than 0.2 m all the time. Finally, the computation time is also one of the necessary performance metrics to judge the effectiveness and efficiency of the proposed method. {\color{black}Here, we use the other three methods as the comparative methods, which are chosen as 


(a). Primal quadratic MIP (PQ-MIP) method: The dual optimization problem~\eqref{eq:dual_prob1} in our proposed method is replaced by the primal optimization problem~\eqref{eq:opt_second_pro_2}, which is solved by MOSEK~\cite{aps2019mosek}.

(b). Non-convex MIP (NC-MIP) method: The problem~\eqref{eq.mpc2} with MIP-based constraints are directly solved in a centralized manner by GUROBI~\cite{gurobi2021gurobi}, i.e., the ADMM-based parallel manner in our proposed method is replaced.

(c). Non-convex quadratically constrained quadratic programming (NC-QCQP) method: The problem~\eqref{eq.multi} with non-convex quadratic distance constraints are solved directly using the log barrier method by the non-convex solver CasADi~\cite{andersson2019casadi}.}

The purpose of this choise and the details of these methods will be introduced in the following content. 

In Fig.~\ref{fig:Trajectory_3D}, the trajectory of each agent is shown in the 3D view. It demonstrates that each agent can track the key point successfully and avoid the collision from each other. Fig.~\ref{fig:Trajectory_2D} presents the 2D views of the trajectory, where the left-top sub-figure shows the top view of the trajectory, and the left-bottom is the side view of the trajectory. Since the single UAV system is an underactuated and second-order nonholonomic system, it is difficult to precisely track the given reference signal with the existing box constraints. Therefore, from Fig.~\ref{fig:Trajectory_3D}, the trajectory of each agent is different from the given reference signal. Both Fig.~\ref{fig:Trajectory_3D} and Fig.~\ref{fig:Trajectory_2D} clearly illustrate the effectiveness of the proposed method.
\textcolor{black}{Additionally, to further showcase the robustness of our proposed method against disturbances, the simulations are implemented under the aforementioned settings, particularly where the UAV models are assumed to be stochastic and contain external disturbance in the system inputs.
	This consideration holds significant practical relevance, further emphasizing the real-world applicability of our methodology.
	In such scenarios, it is expected that a satisfying trajectory can still be computed utilizing our algorithm.
	In the simulation, disturbances with Gaussian distribution (variance is set to be 0.1) are introduced into each system input.
	The results are also depicted in the two figures on the right-hand side of Fig.~{\ref{fig:Trajectory_2D}}.
	It can be evidently shown that the resulting collision-free trajectory can be slightly disturbed, yet all generated trajectories remain satisfying.
	Thus, the simulation results sufficiently demonstrate the robustness of our proposed methodology.}

\begin{figure}[t]
	\centering
	\includegraphics[trim=50 150 100 150, clip, width=0.9\linewidth]{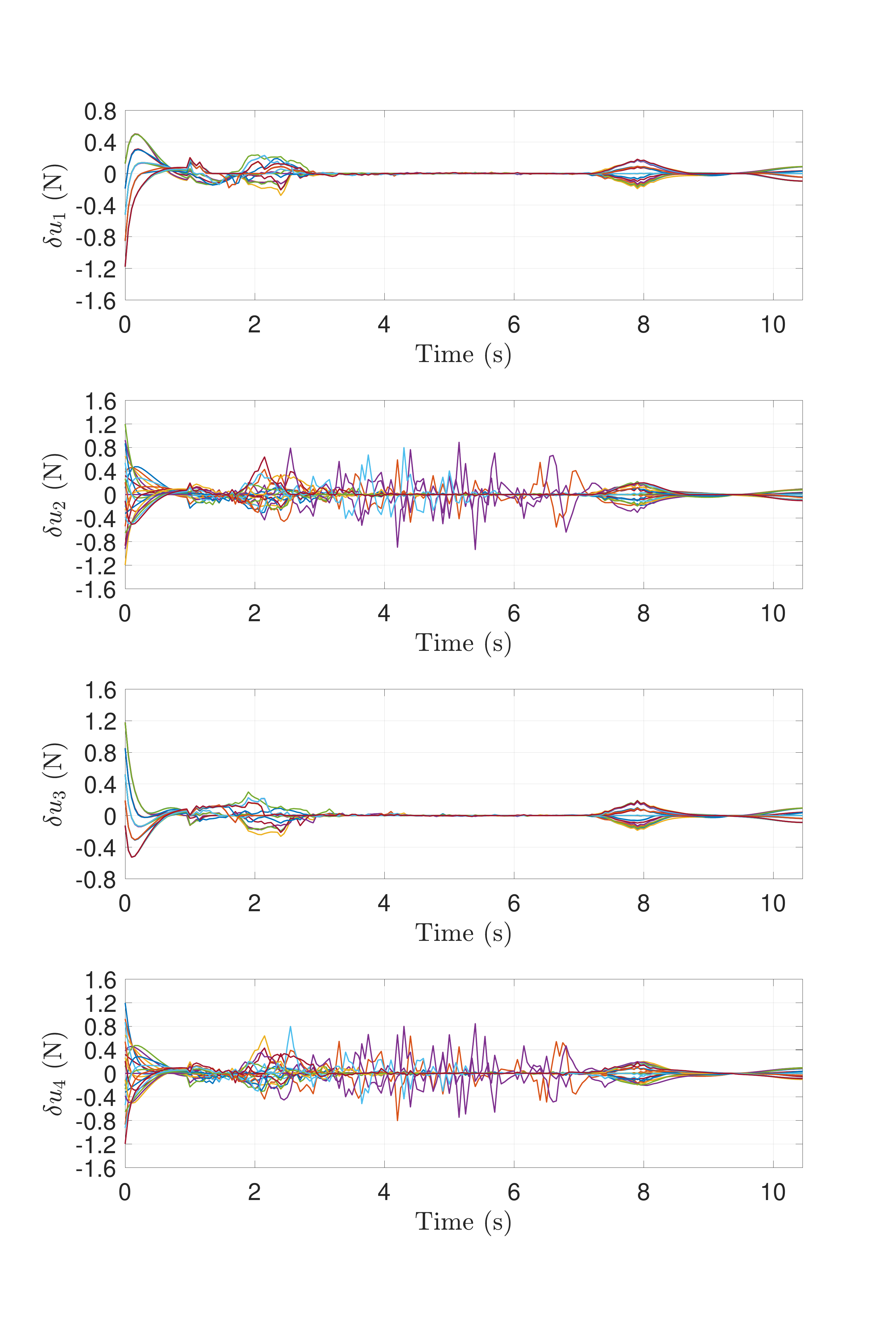}
	\caption{Control inputs of the agents.}
	\label{fig:Control_Inputs}
\end{figure}
To validate the satisfaction of the box constraints, the control input variables for each agent are depicted in Fig.~\ref{fig:Control_Inputs}, where $\delta u_1,\delta u_2,\delta u_3$, and $\delta u_4$ are the control input signals for the four rotors of each agent, respectively. It is clear to see that all control inputs are restricted to the feasible specific range.

\begin{figure}[t]
	\centering
	\includegraphics[trim=70 20 100 50, clip, width=0.8\linewidth]{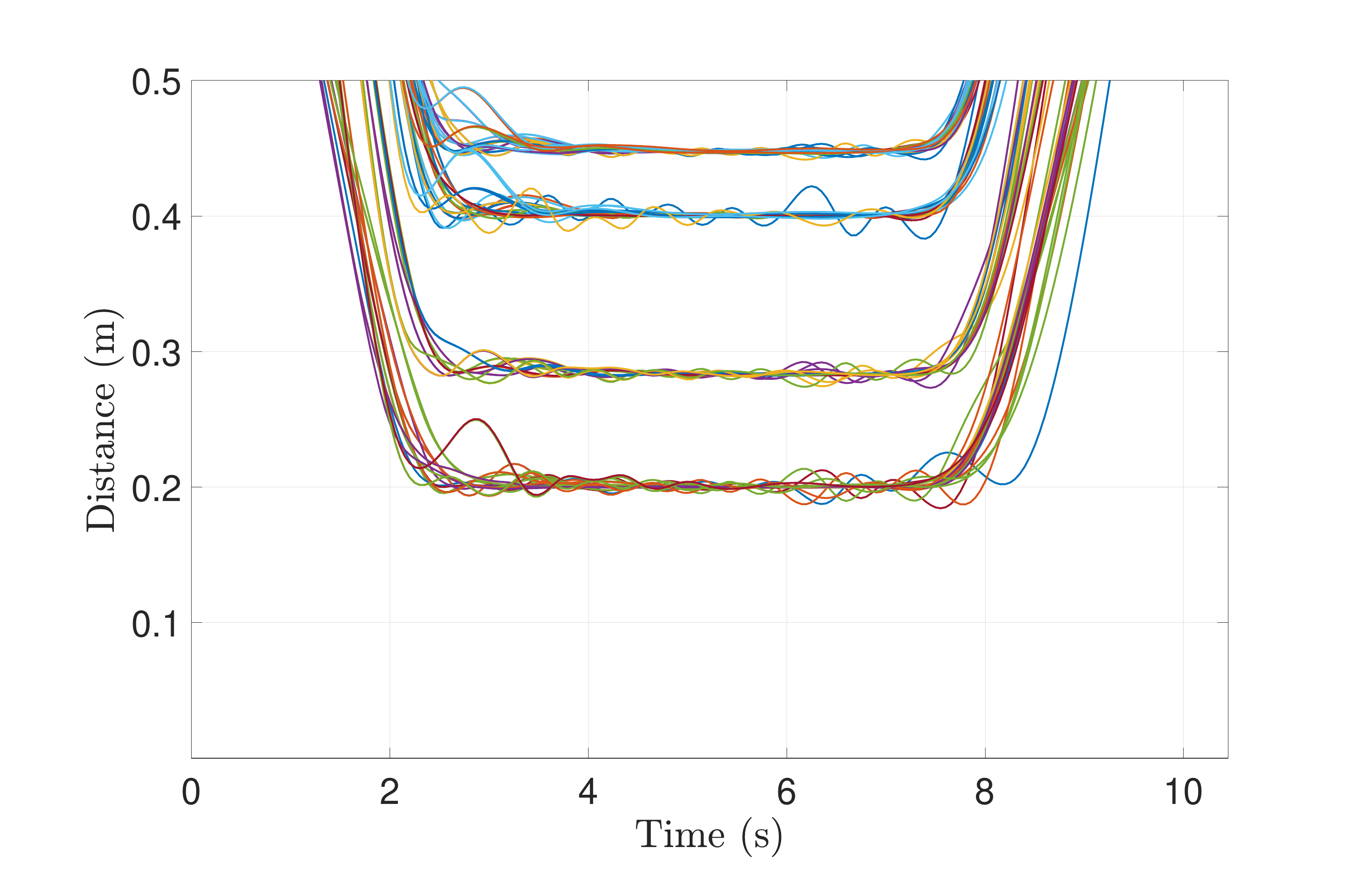}
	\caption{Distances between each pair of agents.}
	\label{fig:Distance}
\end{figure}

Next, the performance of safe distance is analyzed. In Fig.~\ref{fig:Distance}, the distances between each pair of agents are shown. From this figure, we can see that, the safe distances are satisfied at most of time. However, it also shows that the safe distances are slightly violated sometimes. The reason of this phenomenon is that, the stopping criterion we use in this example cannot guarantee the high precision of the solution. As aforementioned, because of the iterative computation manner of ADMM, looser stopping criterion leads to less iterations; as a result, the computation time is reduced while the precision of the solution is diminished. Therefore, the choice of the stopping criterion balances the optimization efficiency and the optimization precision. As a solution, we can reduce the violation by using a more strict stopping criterion. Alternatively, we can choose a larger safe distance to deal with this problem. To show the influence of the stopping criterion, we have performed additional simulations. In the additional simulation, the stopping criterion is chosen as $\hat \epsilon=0.05$, which is half of the original one. Fig.~\ref{fig:Distance2} shows the distances between each pair of agents with the more strict stopping criterion. Compared to Fig.~\ref{fig:Distance}, Fig.~\ref{fig:Distance2} clearly shows that the violation of the safe distance is reduced when a more strict stopping criterion is chosen. \textcolor{black}{ Noted that, all the pertinent simulation videos are
accessible at \protect\url{https://youtu.be/T0C6-W92KXU}.}

\begin{figure}[t]
	\centering
	\includegraphics[trim=70 20 100 50, clip, width=0.8\linewidth]{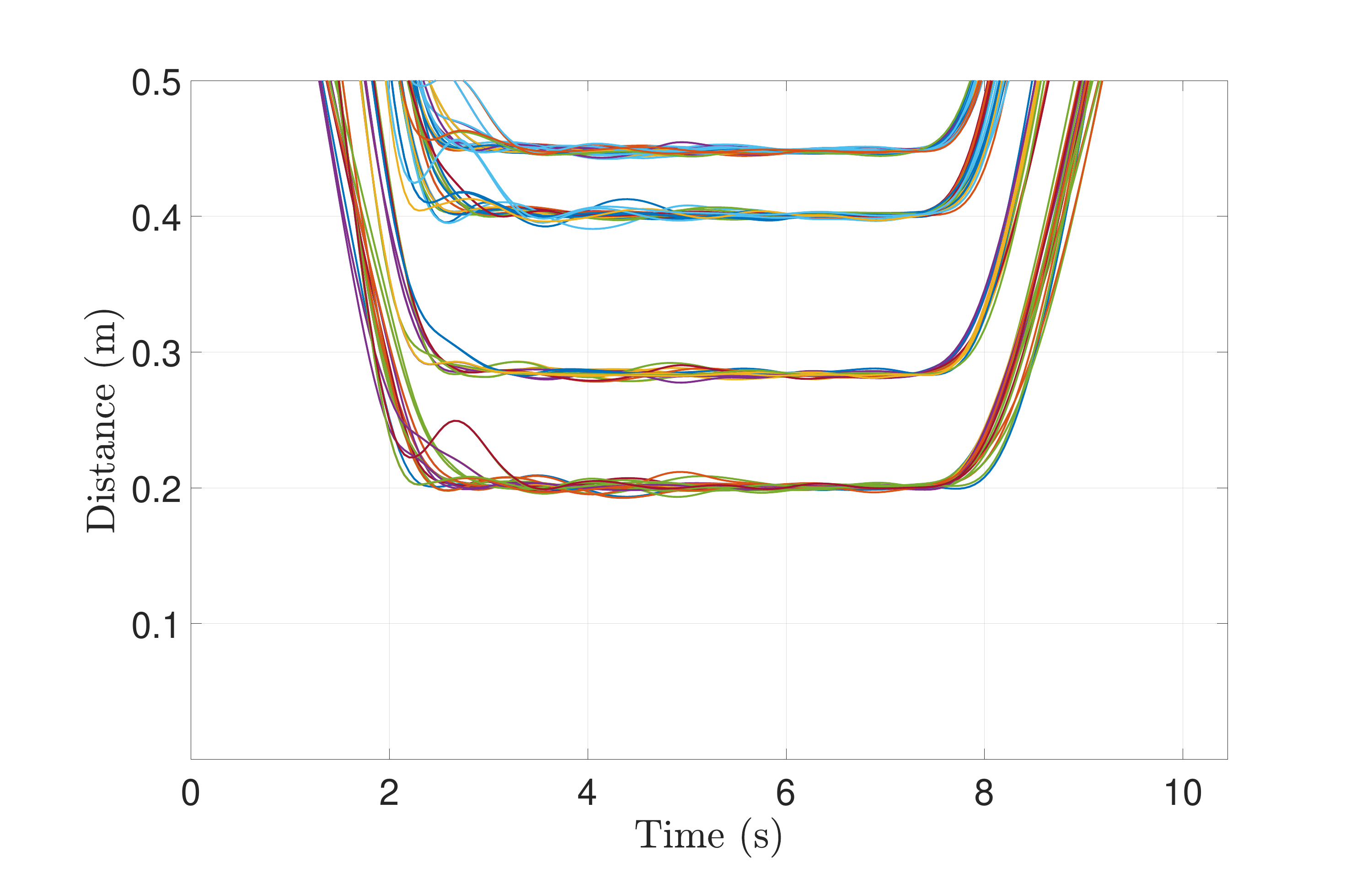}
	\caption{Distances between each pair of agents with the smaller duality gap.}
	\label{fig:Distance2}
\end{figure}

\begin{table}
	\caption{Results of Comparisons (Computation Time)}\label{Table:compare}
\begin{tabular}{|c|c|c|}
	\hline
	& \begin{tabular}[c]{@{}c@{}}Computation Time \\ per Time Stamp (s)\end{tabular} & \multicolumn{1}{c|}{\begin{tabular}[c]{@{}c@{}}Computation Time \\ for ADMM First Iteartion (s)\end{tabular}} \\ \hline
	Proposed Method &            43.4681         &      2.1558             \\ \hline
	PQ-MIP          &            88.6464         &      5.4371             \\ \hline
	NC-MIP          &            181.8962        &      N.A.               \\ \hline
	NC-QCQP         &            441.7536        &      N.A.               \\ \hline
\end{tabular}
\end{table}


\begin{table}
	\caption{\textcolor{black}{Results of Comparisons (Target Positioning Performance)}}\label{Table:compare_motion}
\begin{tabular}{|c|c|c|c|}
	\hline
	& \begin{tabular}[c]{@{}c@{}} Max Endpoint\\ Error  (m) \end{tabular} & \multicolumn{1}{c|}{\begin{tabular}[c]{@{}c@{}}Max Endpoint\\ Error  (m)\end{tabular}} & \multicolumn{1}{c|}{\begin{tabular}[c]{@{}c@{}} Average Endpoint\\ Error  (m)\end{tabular}} \\ \hline
	Proposed Method &             0.1458      &        0.0141     &   0.0848   \\ \hline
	PQ-MIP          &             0.1540        &      0.0675       &   0.0977   \\ \hline
	NC-MIP          &            0.1266       &     0.0661         &  0.0916    \\ \hline
	NC-QCQP         &            73.4317        &      0.4236         &    12.9562  \\ \hline
\end{tabular}
\end{table}

Besides, the computation efficiency of our proposed method is illustrated. Table~\ref{Table:compare} documents the comparisons of computational efficiency using different optimization approaches. Here, we use two indices to judge the computational efficiency, the computational time per time stamp and the computational time for the first ADMM iteration are given. 
It is pertinent to note that, only the proposed method and the PQ-MIP method have the ADMM structure.
\textcolor{black}{The results, as presented in Table II, demonstrate that our proposed method achieves significant improvement in computation efficiency compared to the baselines. Specifically, it reduces computation time for each time stamp by 50.96\%, 76.10\%, and 90.16\% compared to PQ-MIP, NC-MIP, and NC-QCQP, respectively.}
Notice that the computational efficiency of our proposed method is nearly two times higher than the PQ-MIP method, which indicates that solving the dual optimization problem~\eqref{eq:dual_prob1} is more efficient than solving the primal optimization problems~\eqref{eq:opt_second_pro_2}. Additionally, comparing the computation time for ADMM first iteration between these two methods, we can also see that our proposed method is more efficient than the PQ-MIP method. Besides, compared to NC-MIP method, we can conclude that using the ADMM algorithm and the parallel manner improves the efficiency four times than directly solving the non-convex MIP problem. Moreover, comparing between NC-MIP and NC-QCQP, we can infer that using MIP collision avoidance constraints also significantly enhances the computation efficiency. Therefore, as clearly observed from Table~\ref{Table:compare}, our proposed method demonstrates significant superiority over other methods in terms of computational efficiency. \textcolor{black}{Additionally, to further compare our proposed methods with the other comparative methods, the success rate and the motion performance are illustrated and analyzed. For the success rate, it is found that the success rates of all the methods are 100\% after several trials. Then, considering the motion performance as shown in Table~\ref{Table:compare_motion}, it can be seen that the actual endpoint of our proposed method is closer to the target final position. It is also pertinent to note that, the target positioning error of NC-QCQP is significantly large. In fact, the trajectory calculated by NC-QCQP method is actually extremely far from the designed points in this task. One of the main reasons is that, by using log barrier in NC-QCQP, only the constraints in the optimization problem could be satisfied; however, it is unable to guarantee the optimality of the performance (equivalent to the cost value of the corresponding cost function). In conclusion, it can be demonstrated that, our proposed method can achieve higher performance with shorter computation time.}    

\subsection{Discussion on Limitations and Future Works}

After the analysis of the simulation results, it is also pertinent to discuss on the limitations of our proposed method. Since solving the MPC problem highly relies on the accuracy of the model information, the existence of significant model uncertainty, external disturbance, and measurement noise could cause downgrade of the overall performance or even the occurrence of control failures in real-world applications. As a possible solution, the proposed algorithm can be easily combined with several well-known and practical methods to improve the robustness in the scenario with significant disturbance. For example, in the case where the measurement noise is significant, several filters such as the Kalman filter can remarkably improve the measurement accuracy, especially in the case where the Kalman compensator can be modeled and identified at the sensor level. For more sophisticated cases, different kinds of disturbance observers (DOB) [2] can be straightforwardly included to address the disturbance. 
Besides, it is pertinent to note that the potential of the proposed method can be further deployed advantageously to multi-agent systems with larger scale.
Therefore, as also part of our future work, 
we will further investigate and research into 
developments involving fast and efficient 
numerical
computations
(likely involving also up-scale and much more powerful multi-processing hardware);
and extend the proposed method to deal with large-scale multi-agent systems.

\section{Conclusion}
\textcolor{black}{In this work, the collision-free control problem for multi-agent systems is investigated. }
To parallelize the resulting proposed optimization problem in the model predictive collision-free control scheme, the ADMM technique is adroitly utilized to decompose the global problem into sub-problems. Specifically, in the first ADMM iteration, the dual problem of the optimization problem is formulated to enhance the optimization efficiency; and in the second ADMM iteration, a non-convex mixed integer quadratic programming problem is formulated and an intuitive initialization method is presented to accelerate the computational process. Furthermore, rigorous proofs pertaining to the convergence of the proposed non-convex ADMM iterations are furnished. Finally, a multi-agent system with a group of unmanned aerial vehicles (UAVs) is given as an illustrative example to validate the effectiveness of the proposed method. Comparison results clearly demonstrate the promising effectiveness of the proposed methodologies here (involving an innovative parallel computation framework which is constructed, and which facilitates the deployment of advanced optimization methods to solve the multi-agent collision-free control problem).

\bibliographystyle{IEEEtran}
\bibliography{IEEEabrv,Reference}

\end{document}